\documentclass[runningheads,a4paper,orivec,envcountsame,envcountsect,final]{llncs}

\newcommand{\takeout}[1]{\empty}

\usepackage{ifdraft}
\ifdraft{
\usepackage[marginparwidth=2cm,nohead,nofoot,
            headsep = 1cm,
            footskip = 1cm,
            text={12.2cm,19.3cm},
            paperwidth=16.2cm,
            paperheight=23.3cm,
            marginparsep=1mm,
            marginparwidth=1.8cm,
            footskip=1cm,
            centering
            ]{geometry}
}{
  \usepackage[paperheight=235mm,paperwidth=155mm,textwidth=12.2cm,textheight=19.3cm]{geometry}
}

\takeout{ 

\usepackage[inline]{showlabels}
\showlabels{bibitem}

\usepackage{seqsplit}
\usepackage{xstring}
\usepackage{xcolor}
\usepackage{lscape}

\makeatletter

\makeatother
}

\usepackage[notref,notcite]{showkeys}

\usepackage{seqsplit}
\usepackage{xstring}
\newcommand{\defaultshowkeysformat}[1]{%
\StrSubstitute{#1}{ }{\textvisiblespace}[\TEMP]%
\parbox[t]{\marginparwidth}{\raggedright\normalfont\small\ttfamily\(\{\){\color{red!50!black}\expandafter\seqsplit\expandafter{\TEMP}}\(\}\)}%
}

\renewcommand*\showkeyslabelformat[1]{%
\noexpandarg%
\defaultshowkeysformat{#1}%
}

\makeatletter

\renewcommand\section{\@startsection{section}{1}{\z@}%
  {-18\p@ \@plus -4\p@ \@minus -4\p@}%
  {12\p@ \@plus 4\p@ \@minus 4\p@}%
  {\normalfont\large\bfseries\boldmath
    \rightskip=\z@ \@plus 8em\pretolerance=10000 }}
\renewcommand\subsection{\@startsection{subsection}{2}{\z@}%
 {-6\p@ \@plus -2\p@ \@minus -2\p@}%
 {-0.5em \@plus -0.22em \@minus -0.1em}%
 {\normalfont\normalsize\bfseries\boldmath}}                      

\newcommand\mysubsec{\@startsection{paragraph}{4}{\z@}%
  {-6\p@ \@plus -4\p@ \@minus -4\p@}%
  {-0.5em \@plus -0.22em \@minus -0.1em}%
  {\normalfont\normalsize\bfseries}}
\makeatother

\usepackage[utf8]{inputenc}
\usepackage[T1]{fontenc}
\usepackage[british]{babel}
\usepackage{xspace}
\usepackage{dsfont}
\usepackage{tabularx}
\usepackage{dsfont}

\usepackage[inline]{enumitem}
\setlist[enumerate,1]{label=(\arabic*),font=\normalfont,align=left,leftmargin=0pt,labelindent=0pt,listparindent=\parindent,labelwidth=0pt,itemindent=!,topsep=3pt,parsep=0pt,itemsep=3pt,start=1}
\setlist[enumerate,2]{label=(\alph*),font=\normalfont,labelindent=*,leftmargin=*,start=1}
\setlist[itemize]{labelindent=*,leftmargin=*,topsep=5pt,itemsep=3pt}
\setlist[description]{labelindent=*,leftmargin=*,itemindent=-1 em}

\usepackage[noadjust]{cite} 
\usepackage{graphicx,csquotes}
\usepackage[final]{hyperref}
\hypersetup{
  colorlinks=true,
  linkcolor=black,
  citecolor=black,
  filecolor=black,
  urlcolor=black,
  pdftitle={Coalgebraic Language Semantics for Nominal Automata},
  pdfauthor={Florian Frank, Stefan Milius, Henning Urbat},
  pdfkeywords={regular nominal nondeterministic automata, coalgebraic
    language semantics},
  pdfduplex={DuplexFlipLongEdge},
}
  
\ifdraft{%
  \makeatletter
  \setcounter{tocdepth}{2}
  \makeatother
}{}

\usepackage{amssymb,amsxtra,mathrsfs,mathtools,bm,stmaryrd}
\numberwithin{equation}{section}

\usepackage{tikz-cd}

\tikzstyle{shiftarr}=[
        rounded corners,%
        to path={--([#1]\tikztostart.center)
                     -- ([#1]\tikztotarget.center) \tikztonodes
                     -- (\tikztotarget)},
]

\usepackage[footnote,nomargin]{fixme} 
\FXRegisterAuthor{hu}{ahu}{HU} 
\FXRegisterAuthor{sm}{asm}{SM} 
\FXRegisterAuthor{rm}{arm}{RM} 

\usepackage{microtype}
\allowdisplaybreaks

\addto\extrasbritish{
}

\spnewtheorem{assumption}[theorem]{Assumption}{\bfseries}{\rmfamily}
\spnewtheorem{notation}[theorem]{Notation}{\bfseries}{\rmfamily}
\spnewtheorem{observation}[theorem]{Observation}{\bfseries}{\rmfamily}
\spnewtheorem{defn}[theorem]{Definition}{\bfseries}{\rmfamily}
\spnewtheorem{expl}[theorem]{Example}{\bfseries}{\rmfamily}
\spnewtheorem{rem}[theorem]{Remark}{\bfseries}{\rmfamily}
\spnewtheorem{construction}[theorem]{Construction}{\bfseries}{\rmfamily}
\spnewtheorem{examples}[theorem]{Examples}{\bfseries}{\rmfamily}
\spnewtheorem{example_}[theorem]{Example}{\bfseries}{\rmfamily}

\spnewtheorem*{HSP}{General HSP Theorem}{\bf}{\itshape}
\spnewtheorem*{CompThm}{General Completeness Theorem}{\bf}{\itshape}

\usepackage{etoolbox}
\makeatletter
\newenvironment{notheorembrackets}{%
\csdef{@spopargbegintheorem}##1##2##3##4##5{\trivlist%
      \item[\hskip\labelsep{##4##1\ ##2}]{##4{##3}\@thmcounterend\ }##5}%
    }{%
\csdef{@spopargbegintheorem}##1##2##3##4##5{\trivlist%
      \item[\hskip\labelsep{##4##1\ ##2}]{##4(##3)\@thmcounterend\ }##5}%
    }
\makeatother

\renewcommand{\S}{\mathscr{S}}

\newcommand{\Nom}{\mathsf{Nom}}

\newcommand{\mybar}[3]{%
  \mathrlap{\hspace{#2}\overline{\scalebox{#1}[1]{\phantom{\ensuremath{#3}}}}}\ensuremath{#3}
}

\newcommand{\barF}{\mybar{0.6}{2.5pt}{F}}
\newcommand{\barG}{\mybar{0.6}{2pt}{G}}
\newcommand{\barGF}{\mybar{0.85}{2pt}{GF}}
\newcommand{\barFpG}{\mybar{0.9}{2pt}{F + G}}
\newcommand{\barFtG}{\mybar{0.9}{2pt}{F \times G}}
\newcommand{\barAbs}{\mybar{0.8}{1.5pt}{[\names]}}

\newcommand{\hatF}{\widehat{F}}
\newcommand{\hatG}{\widehat{G}}

\newcommand{\lan}{\ddagger}

\newcommand{\dsS}{\mathds{S}}

\newcommand{\midmid}{\hspace{0.2ex}{\rule[-0.1ex]{0.6pt}{1.65ex}}\hspace{0.2ex}}
\newcommand{\scriptmidmid}{\hspace{0.2ex}{\rule[-0.1ex]{0.6pt}{1.1ex}}\hspace{0.2ex}}

\newcommand{\newletter}[1]{{\midmid}#1}
\newcommand{\scriptnew}[1]{{\scriptmidmid}#1}

\newcommand{\DCPOb}{\mathsf{DCPO}_\bot}

\newcommand{\ter}{\tau}
\newcommand{\ini}{\iota}
\newcommand{\set}[1]{\{#1\}}

\newcommand{\cpowfs}{\mathcal Q_{\fs}}
\newcommand{\opp}{\mathsf{op}}

\newcommand{\tr}{\mathsf{tr}}
\newcommand{\inl}{\mathsf{inl}}
\newcommand{\inr}{\mathsf{inr}}
\renewcommand{\rho}{\varrho}
\DeclareMathOperator{\con}{@}

\newcommand{\Set}{\mathbf{Set}}

\newcommand{\Alg}[1]{\mathsf{Alg(}#1\mathsf{)}}
\newcommand{\Coalg}[1]{\mathsf{Coalg(}#1\mathsf{)}}

\newcommand{\Kl}[1]{\mathsf{Kl}(#1)}
\newcommand{\EM}[1]{\mathsf{EM}(#1)}
\newcommand{\C}{\mathscr{C}}

\newcommand{\Y}{\mathscr{Y}}

\DeclareMathOperator{\supp}{\mathsf{supp}}

\newcommand{\Perm}{\mathrm{Perm}}
\newcommand{\names}{\mathbb{A}}
\newcommand{\barA}{{\mybar{0.7}{.75pt}{\names}}} 
\newcommand{\barstr}{\barA^*/\mathord{=_\alpha}}
\newcommand{\barlang}{\powfs(\barstr)}
\newcommand{\datalang}{\powfs(\names^*)}
\newcommand{\compr}{:}
\newcommand{\ufs}{{\mathsf{ufs}}}
\newcommand{\braket}[1]{\langle #1 \rangle}
\newcommand{\fs}{\mathsf{fs}}
\newcommand{\pow}{\mathcal{P}}
\newcommand{\powufs}{\pow_{\ufs}}
\newcommand{\powfs}{\pow_{\fs}}
\newcommand{\fresh}{\mathbin{\#}}

\renewcommand{\epsilon}{\varepsilon}
\newcommand{\eps}{\varepsilon}
\newcommand{\id}{\mathit{id}}
\newcommand{\Id}{\mathsf{Id}}
\newcommand{\seq}{\subseteq}
\newcommand{\xra}[1]{\xrightarrow{~#1~}}
\newcommand{\xla}[1]{\xleftarrow{~#1~}}

\renewcommand{\phi}{\varphi}

\newcommand{\ev}{\mathsf{ev}}

\newcommand{\hookto}{\hookrightarrow}
\newcommand{\subto}{\hookto}
\newcommand{\epito}{\twoheadrightarrow}

\newcommand{\monoto}{\rightarrowtail}

\newcommand{\Pow}{\mathcal{P}}

\newcommand{\xto}[1]{\xra{#1}}

\makeatletter
\newcommand{\xTo}[2][]{\ext@arrow 0359\Rightarrowfill@{#1}{#2}}
\makeatother

\title{Coalgebraic Semantics for Nominal Automata}
\titlerunning{Coalgebraic Semantics for Nominal Automata}

\author{Florian Frank \and Stefan Milius\thanks{Funded by the Deutsche
    Forschungsgemeinschaft (DFG, German Research Foundation) --
    project number 419850228} \and Henning Urbat$^\star$}
\authorrunning{F.~Frank, S.~Milius and H.~Urbat}

\institute{Friedrich-Alexander-Universit\"at Erlangen-N\"urnberg}

\begin{document}
\maketitle

\begin{abstract}
  This paper provides a coalgebraic approach to the language semantics
  of two types of non-deterministic automata over nominal sets:
  non-deterministic orbit-finite automata (NOFAs) and regular nominal
  non-deterministic automata (RNNAs), which were introduced in
  previous work. While NOFAs are a straightforward nominal version of
  non-deterministic automata, RNNAs feature ordinary as well as name
  binding transitions. Correspondingly, words accepted by RNNAs are
  strings formed by ordinary letters and name binding letters. Bar
  languages are sets of such words modulo $\alpha$-equivalence, and to
  every state of an RNNA one associates its accepted bar language. We
  show that the semantics of NOFAs and RNNAs, respectively, arise both
  as an instance of the Kleisli-style coalgebraic trace semantics as
  well as an instance of the coalgebraic language semantics obtained
  via generalized determinization. On the way we revisit coalgebraic
  trace semantics in general and give a new compact proof for the main
  result in that theory stating that an initial algebra for a functor
  yields the terminal coalgebra for the Kleisli extension of the
  functor. Our proof requires fewer assumptions on the functor than
  all previous ones.
\end{abstract}

\section{Introduction}\label{S:intro}

Classical automata and their language semantics have long been
understood in the theory of coalgebras. For example, it is a well-known
exercise~\cite{Rutten98} that standard deterministic automata over a
fixed alphabet can be modelled as coalgebras, that the terminal
coalgebra is formed by all formal languages over that alphabet, and
the unique homomorphism into the terminal coalgebra assigns to each
state of an automaton the language it accepts. Non-deterministic
automata are also coalgebras for a functor extending the one for
deterministic automata in order to accomodate non-deterministic
branching. Their language semantics can be obtained coalgebraically in
two different ways. First, in the \emph{coalgebraic trace semantics}
by Hasuo et al.~\cite{HasuoEA07} one considers coalgebras for composed
functors $TF$ where $F$ is a set functor modelling the type of
transitions and $T$ is a set monad modelling the type of branching;
for example, for non-deterministic branching one takes the power-set
monad. Under certain conditions on $F$ and $T$, including that $F$ has
an extension $\barF$ to the Kleisli category of $T$, an initial
$F$-algebra is seen to lift to the terminal coalgebra for $\barF$. Its
universal property then yields the coalgebraic trace semantics. Among
the instances of this is the standard language semantics of
non-deterministic automata.

Second the \emph{coalgebraic language semantics}~\cite{bms13} is based on
generalized determinization by Silva et al.~\cite{SilvaEA13}. Here one
considers coalgebras for composed functors~$GT$ where $G$ models
transition types and $T$ again models the branching type. Assuming
that $G$ has a lifting to the Eilenberg-Moore category for $T$, generalized
determinization turns such a coalgebra into a $G$-coalgebra by taking
the unique extension of the coalgebra structure to the free
Eilenberg-Moore algebra on the set of states. Moreover, taking the
unique homomorphism from that coalgebra into the terminal
$G$-coalgebra yields the coalgebraic language semantics. In the
leading instance of non-deterministic automata, generalized
determinization is the well-known power-set construction and
coalgebraic language semantics the standard automata-theoretic
language semantics once again.

These two approaches were brought together by Jacobs et
al.~\cite{JacobsEA15} who study those species of systems which can be
modelled as coalgebras in both of the above ways. They show that
whenever there exists an \emph{extension} natural transformation
$TF \to GT$ satisfying two natural equational laws, then the two above
semantics are canonically related, and they agree in the
instances studied in op.~cit.

It is our aim in this paper to draw a similar picture for
non-deterministic automata for languages over infinite alphabets. Such
alphabets allow to model \emph{data}, such as nonces~\cite{KurtzEA07},
object identities~\cite{GrigoreEA13}, or abstract
resources~\cite{CianciaSammartino14}, and the ensuing languages are
therefore called \emph{data languages}. There are several species of
automata for data languages in the literature. We focus on two types
which are known to have a presentation as coalgebras over the category
of nominal sets: non-deterministic orbit-finite automata
(NOFA)~\cite{BojanczykEA14} and regular non-deterministic nominal
automata (RNNA)~\cite{SchroderEA17}. For both of these types of
automata one works with the category of nominal sets and takes the set
of \emph{names} as the alphabet. While NOFAs are a straightforward
nominal version of standard non-deterministic automata, RNNAs feature
\emph{binding transitions}, which can be thought as storing an input
name in a `register' for comparison with future input
names. Correspondingly, they accept words including name binding
letters and which are taken modulo $\alpha$-equivalence; such words
form \emph{bar languages} (the name stems from the bar in front of
name binding letters $\newletter a$). However, while these automata
are understood as coalgebras, their semantics has not been studied
from a coalgebraic perspective so far.

We fill this gap here and prove that the data language accepted by a NOFA and the bar
language accepted by an RNNA arise as instances of both coalgebraic trace semantics
\renewcommand{\theoremautorefname}{Theorems}%
(\autoref{T:NOFA-Kl} and~\ref{T:RNNA-Kl})
\renewcommand{\theoremautorefname}{Theorem}%
and coalgebraic language semantics
\newcommand{\corollaryautorefname}{Corollaries}%
(\autoref{C:NOFA-EM} and~\ref{C:RNNA-EM}).
\renewcommand{\corollaryautorefname}{Corollary}%
The latter result is obtained by using canonical extension natural
transformations obtained from the result by Jacobs et al.~\cite{JacobsEA15}. 

While these results will perhaps hardly surprise the cognoscenti, and
the treatment of NOFAs indeed appears as an(other) exercise in coalgebra, we
should like to point out that there are a number of technical
subtleties arising in the treatment of RNNAs. Essentially, what causes
some trouble is the presence of the abstraction functor in their type. We
solve all these difficulties by working with the uniformly finitely
supported power-set monad $\powufs$ on nominal sets in lieu of the
more common finitely supported power-set monad $\powfs$ (which provides
the power objects of the topos of nominal sets). Note also that for a
nominal set $X$, neither $\powfs X$ nor $\powufs X$ form cpos (so, in
particular, they do not form complete lattices). Hence, it may come as
a bit of a surprise that the Kleisli categories of both monads are
nevertheless enriched over complete lattices (\autoref{P:dcpo}), one of
the key requirements for coalgebraic trace semantics.

We present our results in a modular way so that they may be reusable
for the study of coalgebraic semantics for other types of nominal
systems, such as nominal tree automata. For example, we show that all
\emph{binding polynominal functors}, e.g.~those functors arising from
a binding signature in the sense of Fiore et al.~\cite{FioreEA99} have a
canonical extension to the Kleisli category of $\powufs$
(\autoref{cor:extension-pufs}). Analogously, we show a lifting result
for terminal coalgebras to the Eilenberg-Moore category for a subclass
of these functors (\autoref{C:lift-nu}).

Last but not least, on the way to the coalgebraic semantics of NOFAs
and RNNAs we take a fresh look at coalgebraic trace semantics in
general. We provide a new compact proof for the main theorem of that
theory. It states that for a functor $F$ and a monad $T$ satisfying
certain conditions, including that $F$ has an extension $\barF$ to the
Kleisli category of $T$, the initial $F$-algebra extends to a terminal
coalgebra for $\barF$ (\autoref{T:Kl}). We obtain this essentially as
a combination of Hermida and Jacobs' adjoint lifting
theorem~\cite[Thm.~2.14]{HermidaJ98} and an argument originally given
by Freyd~\cite{Freyd92} that for locally continuous endofunctors on
categories enriched in cpos an initial algebra yields a terminal
coalgebra. Here we adjust this argument to work for locally monotone
endofunctors on categories enriched in directed-complete partial
orders. As a consequence, our proof does not require the existence of
a zero object in the Kleisli category of $T$ and, notably, we only
need the mere existence of the initial algebra for $F$ and not that it is
obtained after $\omega$ steps of the initial-algebra chain given by
$F^n 0$ ($n <\omega$).


\section{Preliminaries}\label{S:prelim}

\subsection{Nominal Sets}
\hunote[inline]{Copy \& paste from our CONCUR paper.}

Nominal sets form a convenient formalism for dealing with names and
freshness; for our present purposes, names play the role of data. We
briefly recall basic notions and facts and refer to Pitts'
book~\cite{Pitts13} for a comprehensive introduction. Fix a countably
infinite set $\names$ of \emph{names}, and let $\Perm(\names)$ denote
the group of finite permutations on $\names$, which is generated by
the \emph{transpositions} $(a\, b)$ for $a\neq b\in\names$ (recall
that $(a\, b)$ just swaps~$a$ and~$b$). A \emph{nominal set} is a
set~$X$ equipped with a (left) group action
$\Perm(\names)\times X\to X$, denoted $(\pi,x)\mapsto \pi\cdot x$,
such that every element $x\in X$ has a finite
\emph{support}~$S\subseteq\names$, i.e.~$\pi\cdot x=x$ for every
$\pi\in \Perm(\names)$ such that $\pi(a)=a$ for all $a\in S$. Every
element~$x$ of a nominal set $X$ has a least finite support, denoted
$\supp(x)$. Intuitively, one should think of $X$ as a set of syntactic
objects (e.g.~strings, $\lambda$-terms, programs), and of $\supp(x)$
as the set of names needed to describe an element $x\in X$. A name
$a\in\names$ is \emph{fresh} for~$x$, denoted $a\fresh x$, if
$a\notin\supp(x)$. The \emph{orbit} of an element $x\in X$ is given by
$\{ \pi\cdot x: \pi\in\Perm(\names)\}$. The orbits form a partition of
$X$. The nominal set $X$ is \emph{orbit-finite} if it has only
finitely many orbits.

A map $f\colon X\to Y$ between nominal sets is \emph{equivariant} if
$f(\pi\cdot x)=\pi\cdot f(x)$ for all $x\in X$ and
$\pi\in \Perm(\names)$. Equivariance implies
$\supp(f(x))\seq \supp(x)$ for all $x\in X$. We denote by $\Nom$ the category of nominal sets and equivariant
maps.

Putting $\pi\cdot a = \pi(a)$ makes $\names$ into a nominal
set. Moreover,~$\Perm(\names)$ acts on subsets $A\subseteq X$ of a
nominal set~$X$ by $\pi\cdot A = \{\pi\cdot x \compr x \in A\}$. A
subset $A\seq X$ is \emph{equivariant} if $\pi\cdot A=A$ for all
$\pi\in \Perm(\names)$. More generally, it is \emph{finitely
  supported} if it has finite support w.r.t.\ this action, i.e.~there
exists a finite set $S\seq \names$ such that $\pi\cdot A = A$ for all
$\pi\in \Perm(\names)$ such that $\pi(a)=a$ for all $a\in S$. The set
$A$ is \emph{uniformly finitely supported} if
$\bigcup_{x\in A} \supp(x)$ is a finite set. This implies that $A$ is
finitely supported, with least support
$\supp(A)=\bigcup_{x\in A}
\supp(x)$~\cite[Theorem~2.29]{gabbay2011}. (The converse does not
hold, e.g.~the set $\names$ is finitely supported but not uniformly
finitely supported.) Uniformly finitely supported orbit-finite sets
are always finite (since an orbit-finite set contains only finitely
many elements with a given finite support).  We denote by $\Pow_\ufs\colon \Nom\to \Nom$ and
$\powfs\colon \Nom\to \Nom$ the endofunctors sending a nominal set $X$ the its set of (uniformly) finitely supported subsets and an equivariant map
$f\colon X\to Y$ to the map $A\mapsto f[A]$. 

The coproduct $X+Y$ of nominal sets $X$ and $Y$ is given by their
disjoint union with the group action inherited from the two
summands. Similarly, the product $X \times Y$ is given by the
cartesian product with the componentwise group action; we have
$\supp(x,y) = \supp(x) \cup \supp(y)$. Given a nominal set $X$
equipped with an equivariant equivalence relation, i.e.~an equivalence
relation $\sim$ that is equivariant as a subset
$\mathord{\sim} \subseteq X \times X$, the quotient $X/\mathord{\sim}$
is a nominal set under the expected group action defined by
$\pi \cdot [x]_\sim = [\pi \cdot x]_\sim$.

A key role in the theory of nominal sets is played by
\emph{abstraction sets}, which provide a semantics for binding
mechanisms~\cite{GabbayPitts99}. Given a nominal set $X$, an equivariant equivalence relation $\sim$ on
  $\names \times X$ is defined by $(a,x)\sim (b,y)$ iff
  $(a\, c)\cdot x=(b\, c)\cdot y$ for some (equivalently, all)
  fresh~$c$. The \emph{abstraction set} $[\names]X$ is the quotient
  set $(\names\times X)/\mathord{\sim}$. The $\sim$-equivalence class
  of $(a,x)\in\names\times X$ is denoted by
  $\braket{a} x\in [\names]X$. We may think of~$\sim$ as an abstract notion of $\alpha$-equivalence,
and of~$\braket{a}$ as binding the name~$a$. Indeed we have
$\supp(\braket{a} x)= \supp(x)\setminus\{a\}$ (while
$\supp(a,x)=\{a\}\cup\supp(x)$), as expected in binding constructs.

 The object map $X\mapsto [\names]X$ extends to an endofunctor
$[\names]\colon \Nom \to \Nom$
sending an equivariant map $f\colon X\to Y$ to the equivariant map $[\names]f\colon [\names]X\to [\names]Y$ given by $\braket{a}x\mapsto \braket{a}f(x)$ for $a\in \names$ and $x\in X$.

\subsection{Nominal Automata}\label{sec:nom-aut}
In this section, we recall two notions of nominal automata earlier
introduced in the literature: non-deterministic orbit-finite automata
(NOFAs)~\cite{BojanczykEA14} and regular non-deterministic nominal
automata (RNNAs)~\cite{SchroderEA17}. The former accept \emph{data
  languages} (consisting of finite words over an infinite alphabet)
while the latter accept \emph{bar languages} (consisting of finite
words formed by ordinary letters and name binding ones, taken modulo
$\alpha$-equivalence).

\begin{notheorembrackets}
  \begin{defn}[\cite{BojanczykEA14}]
    (1)~A \emph{NOFA} $A=(Q,R,F)$ is given by an orbit-finite nominal set $Q$
  of \emph{states}, an equivariant relation
  $R\seq Q\times \names \times Q$ specifying \emph{transitions}, and
  an equivariant set $F\seq Q$ of
  \emph{final states}. We write $q \xra{a} q'$ in lieu of $(q,a,q')
  \in R$.
  
  \begin{enumerate}\stepcounter{enumi}
\item  Given a string
    $w=a_1a_2\cdots a_n\in \names^*$ and a state $q\in Q$, a \emph{run} for
    $w$ from $q$ is a sequence of transitions 
    \(
      q\xto{a_1}q_1\xto{a_2}\cdots \xto{a_n}q_n.
    \)
    The run is \emph{accepting} if $q_n$ is final. The state $q$
    \emph{accepts} $w$ if there exists an accepting run for $w$ from
    $q$. The data language \emph{accepted} by $q$ is given by $\{ w\in \names^*: \text{$q$ accepts $w$}\}$.
  \end{enumerate}
\end{defn}
\end{notheorembrackets}
NOFAs are known to be expressively equivalent to \emph{finite memory
  automata}~\cite{KaminskiFrancez94}. We note that in contrast to \cite{BojanczykEA14} we do not require NOFAs to have an initial state $q_0\in Q$; this is more natural from a coalgebraic point of view. Moreover, the orbit-finiteness of the states is not relevant for our results and could be dropped.

\begin{rem}\label{rem:nofa-as-coalgebras}
\begin{enumerate}
\item%
  \smnote{In general, we should not start sentences with `recall' if
    we don't say \emph{from where} one should this recall; this is the
    arrogant mathematicians way of assuming s.th.~about the knowledge
    of the reader.}  Given an endofunctor $F$ on a category $\C$, an
  \emph{$F$-coalgebra} is a pair $(C,c)$ of an object $C$ and a
  morphism $c\colon C\to FC$ on $\C$. A \emph{homomorphism} of
  $F$-coalgebras from
  $(C,c)$ to $(D,d)$ is a morphism $h\colon C\to D$ with $d\cdot h = Fh\cdot c$.
  
\item A NOFA corresponds precisely to an
  orbit-finite coalgebra
  \(
    \langle f, \delta\rangle\colon Q\longrightarrow 2\times \powfs(\names\times Q)
  \)
  for the functor on $\Nom$ given by
  \[
    Q\mapsto \powfs(1+\names\times Q)\cong 2\times \powfs(\names\times
    Q).
  \]
  In fact, $f\colon Q\to 2$ defines the equivariant set $F\seq Q$ of
  final states and $\delta\colon Q\to \powfs(\names\times Q)$ defines
  the transitions via
  $q\xto{a}q'$ iff $(a,q')\in \delta(q)$.
\end{enumerate}
\end{rem}

In order to incorporate explicit name binding into the
automata-theoretic setting, we work with \emph{bar strings},
i.e. finite words over the infinite alphabet
\[
  \barA := \names \cup \{\newletter a: a\in \names\}.
\]
We denote the nominal set of all bar strings by $\barA^*$, and we
equip it with the group action defined pointwise. The letter
$\newletter a$ is interpreted as binding the name~$a$ to the
right. Accordingly, a name $a\in \names$ is said to be \emph{free} in
a bar string $w\in \barA^*$ if (1)~the letter $a$ occurs in $w$, and
(2)~the first occurrence of $a$ is not preceded by any occurrence of
$\newletter a$. For instance, the name $a$ is free in
$a\newletter aba$ but not free in $\newletter aaba$, while the name
$b$ is free in both bar strings. This yields a natural notion of
$\alpha$-equivalence:

\begin{defn}[$\alpha$-equivalence]\label{def:alpha-fin}
  Let $=_\alpha$ be the least equivalence relation on~$\barA^*$ such
  that
  $x\newletter av =_\alpha x\newletter bw$ 
  for all $a,b\in \names$ and $x,v,w\in \barA^*$ such that
      $\braket{a} v = \braket{b} w$.
  We denote by $\barstr$ the sets of $\alpha$-equivalence classes of
   bar strings, and we write $[w]_\alpha$ for the
  $\alpha$-equivalence class of $w\in \barA^*$.
\end{defn}
\begin{rem}\label{rem:alpheq-technical}
  \begin{enumerate}
  \item By Pitts~\cite[Lem.~4.3]{Pitts13}, for every pair
    $v,w\in \barA^*$ the condition $\braket{a} v = \braket{b} w$ holds
    if and only if
    \[
      \text{$a=b$ and $v=w$,}
      \qquad \text{or}\qquad
      \text{$b\fresh v$ and $(a\, b)\cdot v = w$.}
    \]
  \item The equivalence relation $=_\alpha$ is equivariant. Therefore,
    $\barA^*/{=_\alpha}$ forms a nominal set with the group
    action $\pi\cdot [w]_\alpha = [\pi\cdot w]_\alpha$ for
    $\pi\in \Perm(\names)$ and $w\in \barA^*$.  The least support of
    $[w]_\alpha$ is the set of free names of $w$.
  \end{enumerate}
\end{rem}
\makeatletter
\newcommand\gobblepars{%
    \@ifnextchar\par%
        {\expandafter\gobblepars\@gobble}%
        {}}
\makeatother
\begin{notheorembrackets}
  \begin{defn}[\cite{SchroderEA17}]%
    (1)~An \emph{RNNA} $A=(Q,R,F)$ is given by an orbit-finite
    nominal set $Q$ of \emph{states}, an equivariant relation
    $R\seq Q\times \barA \times Q$ specifying \emph{transitions}, and an equivariant set $F\seq Q$
    of \emph{final states}. We write $q\xto{\sigma}q'$ if
    $(q,\sigma,q')\in R$. The transitions are subject to two
    conditions:
    \begin{enumerate}[label=(\alph*)]
    \item \emph{$\alpha$-invariance}: if $q\xto{\scriptnew a}q'$ and
      $\braket{a} q'=\braket{b} q''$, then $q\xto{\scriptnew b}q''$.
        
    \item \emph{Finite branching up to $\alpha$-invariance:} For every
      $q\in Q$ the sets
      \[
        \set{(a,q') \compr q\xto{a}q'}
        \qquad\text{and}\qquad
        \set{\braket{a} q' \compr q\xto{\scriptnew a} q'}
      \]
      are finite (equivalently, uniformly finitely supported).
    \end{enumerate}
    
    \begin{enumerate}\stepcounter{enumi}
    \item Given a bar string
      $w=\sigma_1\sigma_2\cdots \sigma_n\in \barA^*$ and a state
      $q\in Q$, a \emph{run} for $w$ from $q$ is a sequence of
      transitions
      \(
        q\xto{\sigma_1}q_1\xto{\sigma_2}\cdots \xto{\sigma_n}q_n.
      \)
      The run is \emph{accepting} if $q_n$ is final. The state $q$
      \emph{accepts} $w$ if there exists an accepting run for $w$ from
      $q$. The bar language \emph{accepted} by $q$ is given
      by $\{ [w]_\alpha \compr w\in \barA^*, \, \text{$A$ accepts $w$} \}$.
    \end{enumerate}
  \end{defn}
\end{notheorembrackets}

\begin{rem}\label{rem:rnna-as-coalgebras}
  \begin{enumerate}
    As for NOFAs, we do not equip RNNAs with explicit initial states. Similar to \autoref{rem:nofa-as-coalgebras}, RNNAs are seen to correspond to
    coalgebras
    \(
      \langle f, \delta, \tau\rangle\colon Q\longrightarrow 2\times \Pow_\ufs(\names\times Q)\times
      \Pow_\ufs([\names]Q)
    \)
    for the functor on $\Nom$ given by
    \[
      Q\mapsto \powufs(1+\names\times Q+[\names]Q)
      \cong
      2\times \powufs (\names\times Q) \times \powufs([\names] Q).
    \]
    Here $f$ and $\delta$ correspond to final states and free
    transitions, and the equivariant map
    $\tau\colon Q\to \powufs([\names]Q)$ defines the $\alpha$-invariant
    bound transitions via
    $q\xto{\scriptnew{a}} q'$ iff $\braket{a}q'\in \tau(q)$.
    The use of $\powufs$ (in lieu of $\powfs$) ensures that if $Q$ is
    orbit-finite, then the finiteness conditions in the definition of
    an RNNA are met.

However, we note that while our results on coalgebraic semantics are stated for RNNAs they actually hold without orbit-finiteness assumptions.
  \end{enumerate}
\end{rem}

Our goal is to interpret the above ad-hoc definition of the data
languages of a NOFA and the bar languages of an RNNA within the
coalgebraic framework. 

%

\subsection{Initial algebras in $\DCPOb$-enriched categories}
\label{S:DCPOb}

For the Kleisli-style coalgebraic trace semantics we shall make use of
a result which shows that in categories where the hom-sets are
enriched over directed-complete partial orders, the initial algebra and
terminal coalgebra coincide.

Recall that a subset $D \subseteq P$ of a poset $P$ is \emph{directed}
if every finite subset of $D$ has an upper bound in $D$; equivalently,
$D$ is nonempty and for every $x,y \in D$, there exists a $z \in D$
with $x,y \leq z$. The poset $P$ is a \emph{dcpo with bottom} if it
has a least element and \emph{directed joins}, that is, every directed
subset has a join in $P$. We write $\DCPOb$ for the category of dcpos
with bottom and continuous maps between them; a map is
\emph{continuous} if it is monotone and preserves directed joins.

\begin{defn}
  \begin{enumerate}
  \item A category $\C$ is \emph{left strictly $\DCPOb$-enriched} provided that each
    hom-set is equipped with the structure of a dcpo with bottom, and
    composition preserves bottom on the left and is \emph{continuous}: for every morphism
    $f$ and appropriate directed sets of morphisms $g_i$ ($i \in D$)
    we have
    \[
      \textstyle
      \bot \cdot f = \bot, \qquad
      f \cdot \bigvee_{i\in D} g_i = \bigvee_{i\in D} f \cdot g_i,\qquad
      \big(\bigvee_{i \in D} g_i\big) \cdot f = \bigvee_{i\in D} g_i
      \cdot f.
    \]
  \item A functor on $\C$ is \emph{locally monotone} if its
    restrictions $\C(A,B) \to \C(FA,FB)$ to the hom-sets are monotone.
  \end{enumerate}
\end{defn}
\begin{notheorembrackets}
\begin{theorem}[{\cite[Prop.~5.6]{amm21}}]\label{T:DCPOb}\label{T:dcpo}
  %
  Let $F$ be a locally monotone functor on a left strictly $\DCPOb$-enriched category.
  If an initial algebra $(\mu F, \ini)$ exists, then $(\mu F, \ini^{-1})$ is a
  terminal coalgebra.
\end{theorem}
\end{notheorembrackets}
\noindent
(This uses that the structure $\iota\colon F(\mu F) \to \mu F$ of the
initial algebra is an isomorphism by Lambek's Lemma~\cite{Lambek68}.)
This result is an adaptation of an earlier related result proved by
Freyd~\cite{Freyd92} for locally continuous functors on
$\omega$-cpo-enriched categories. Note that preservation of bottom on
the right ($f \cdot \bot = \bot$) is not needed for this result.

\section{Coalgebraic Trace Semantics}
\label{S:trace}

In this section we shall see that the (bar) language semantics of NOFAs
and RNNAs is an instance of coalgebraic trace semantics. To this end we
first adapt and generalize the coalgebraic trace semantics for set
functors by Hasuo et al.~\cite{HasuoEA07} to arbitrary
categories. Here one considers coalgebras for composed functors~$TF$,
where $T$ is a monad modelling a branching type like non-determinism
or probabilistic branching, and $F$ models the type of transitions of
systems. We then instantiate this to coalgebras in $\Nom$ for functors
$TF$, where $T$ is $\powfs$ and $F$ a polynominal functor or
$T=\powufs$ and $F$ a binding polynomial functor. Specifically, we
obtain the two desired types of nominal automata as instances. 

\subsection{General Coalgebraic Trace Semantics Revisited}
We begin by recalling a few facts about extensions of functors to
Kleisli categories.
\begin{rem}\label{R:Kleisli}
  Let $F$ be a functor and $(T, \eta, \mu)$ a monad, both on
  the category~$\C$.
  \begin{enumerate}
  \item\label{R:Kleisli:1} The \emph{Kleisli category} $\Kl T$ has the
    same objects as $\C$ and a morphisms $f$ from $X$ to $Y$ is a
    morphism $f\colon X \to TY$ of $\C$. The composition of $f$ with
    $g\colon Y \to TZ$ is defined by $\mu_Z \cdot Tg \cdot f$ and the
    identity on $X$ is $\eta_X \colon X \to TX$. We have the
    identity-on-objects functor $J\colon \C \to \Kl T$ defined by 
    $J(f\colon X \to Y)  = \eta_Y \cdot f$.

  \item An endofunctor $\barF\colon \Kl T \to \Kl T$ \emph{extends}
    the functor $F$ if $\bar F J = J F$. It is well known and easy to
    prove (see Mulry~\cite{Mulry94}) that extensions of $F$ to $\Kl T$ are in
    bijective correspondence with \emph{distributive laws} of $F$ over
    $T$; these are natural transformations $\lambda\colon FT \to
    TF$ compatible with the monad structure of~$T$:
    \[
      \begin{tikzcd}
        F
        \arrow{r}{F\eta}
        \arrow{rd}[swap]{\eta F}
        &
        FT
        \arrow{d}{\lambda}
        \\
        &
        TF
      \end{tikzcd}
      \qquad\qquad
      \begin{tikzcd}
        FTT
        \arrow{r}{\lambda T}
        \arrow{d}[swap]{F\mu}
        &
        TFT
        \arrow{r}{T\lambda}
        &
        TTF
        \arrow{d}{\mu F}
        \\
        FT
        \arrow{rr}{\lambda}
        &&
        TF
      \end{tikzcd}
    \]

  \item\label{R:Kleisli:3} Let $G$ be a quotient functor of $F$, which means that we have
    a natural transformation with epimorphic components
    $q\colon F \epito G$. Suppose that $F$ extends to $\Kl T$
    via a distributive law $\lambda\colon FT \to TF$. Then an 
    object-indexed family of morphisms $\rho_X\colon GTX \to TGX$ is a
    distributive law of $G$ over $T$ provided that the following
    squares commute
    \[
      \begin{tikzcd}
        FTX
        \ar{r}{\lambda_X}
        \ar[->>]{d}[swap]{q_{TX}}
        &
        TFX
        \ar[->>]{d}{Tq_X}
        \\
        GTX
        \ar{r}{\rho_X}
        &
        TGX
      \end{tikzcd}
      \qquad
      \text{for every object $X$ of $\C$.}
    \]
  \end{enumerate}
\end{rem}
\begin{expl}\label{E:ext}
  \begin{enumerate}
  \item\label{E:ext:1} Constant functors and the identity functor on $\C$ obviously
    extend to $\Kl T$.

  \item For a pair $F, G$ of endofunctors which extend to $\Kl T$,
    their composition extends, too, and we have $\barGF = \barG \barF$.
    
  \item\label{E:ext:2} Suppose that $\C$ has coproducts. Then
    $\barFpG = \bar F + \barG$, for a pair $F,G$ of endofunctors which
    extend to $\Kl \C$. Indeed, for a coproduct $F+G$ one uses that
    $J\colon \C \to \Kl T$, being a left adjoint, preserves
    coproducts. Given extensions $\barF$ and~$\barG$, it is then clear
    that $\barF + \barG$ extends $F+G$:%
    \smnote{We will need the definition on morphisms in the proof of
      \autoref{T:RNNA-Kl}; so we need to provide more details.}  for
    every morphism $f\colon X \to TY$ in $\Kl T$ one has
    \[
      \barFpG (f) = \big(FX + FY \xra{\bar Ff + \bar Gf} TFY+TGY
      \xra{[T\inl, T\inr]} T(FY+GY)\big),
    \]
    where $FY \xra{\inl} FY + GY \xla{\inr} GY$ are the coproduct
    injections. This works similarly for arbitrary coproducts.

  \item\label{E:ext:3} Suppose that $\C$ has finite products. Then
    finite products of functors with an extension can be extended when
    the monad $T$ is \emph{commutative}; this notion was introduced by
    Kock~\cite[Def.~3.1]{Kock70}. It is based on the notion of a
    \emph{strong} monad, that is a monad $T$ equipped with a natural
    transformation $s_{X,Y}\colon X \times TY \to T(X\times Y)$
    (called \emph{strength}) satisfying four natural equational laws
    (two w.r.t.~$1$ and~$\times$ on $\C$ and two w.r.t.~the monad
    structure). We do not recall these laws explicitly since they are
    not needed for our exposition. A strength gives rise to a
    \emph{costrength} $t_{X,Y}\colon TX \times Y \to T(X\times Y)$
    defined by
    \[
      t_{X,Y} = \big( TX\times Y\cong Y \times TX \xra{s_{Y,X}} T(Y
      \times X) \xra{T(\cong)} T(X\times Y) \big).
    \]
    The monad $T$ is \emph{commutative} if the following diagram 
    commutes:
    \[
      \begin{tikzcd}[row sep = 0, column sep = 30]
        &
        T(TX \times Y)
        \ar{r}{Tt_{X,Y}}
        &
        TT(X\times Y)
        \ar{rd}[near start]{\mu_{X\times Y}}
        \\
        TX \times TY
        \ar{ru}[near end]{s_{TX,Y}}
        \ar{rd}[near end, swap]{t_{X,TY}}
        \ar[dashed]{rrr}{d_{X,Y}}
        &&&
        T(X\times Y)
        \\
        &
        T(X \times TY)
        \ar{r}{Ts_{X,Y}}
        &
        TT(X \times Y)
        \ar{ru}[near start,swap]{\mu_{X\times Y}}
      \end{tikzcd}
    \]
    The ensuing natural transformation $d$ in the middle is used to
    extend the product $F \times G$ of endofunctors on $\C$
    having extensions $\barF$ and $\barG$ on $\Kl T$: for every
    morphism $f\colon X \to TY$ in $\Kl T$ one puts%
    \smnote{Note that the barF and barG macros don't work on the label
      arrows; TODO: perhaps fix this later.}
    \[
      \barFtG (f) = \big(FX \times GX \xra{\bar F f \times \bar G f} TFY
      \times TGY \xra{d_{FY,GY}} T(FY \times GY)\big).
    \]

  \end{enumerate}
\end{expl}
\begin{rem}
  \begin{enumerate}
  \item Every set monad is strong via a canonical strength; this
    follows, for example, from Moggi's
    result~\cite[Thm.~3.4]{Moggi91}.
    For example, the power-set functor $\pow\colon \Set \to \Set$ is
    commutative via its canonical strength 
    \begin{equation}\label{eq:pstr}
      s_{X,Y}\colon X \times \pow Y \to \pow(X \times Y)
      \quad
      \text{defined by}
      \quad
      (x,S) \mapsto \set{(x,s) : s \in S}. 
    \end{equation}
  \item As a consequence of what we saw in \autoref{E:ext} every
    polynomial set functor has a canonical extension to the Kleisli
    category of any commutative set monad
    (cf.~\cite[Lem.~2.4]{HasuoEA07}).
    
  \item More generally, this results holds for analytic set
    functors~\cite[Thm.~2.9]{mps09}. That notion was introduced by
    Joyal~\cite{Joyal81,Joyal86}, and he proved that analytic set
    functors are precisely those set functors which weakly preserve wide
    pullbacks.

  \end{enumerate}
\end{rem}

With the help of Hermida and Jacobs' result~\cite[Thm.~2.14]{HermidaJ98} on
extending adjunctions to categories of algebras one easily obtains the
following extension result for initial algebras:
\begin{proposition}\label{P:ini}
  Let $T$ be a monad on the category $\C$ and let $F\colon \C \to \C$
  have an extension $\bar F$ on $\Kl T$. If $(\mu F, \ini)$ is an
  initial $F$-algebra, then $\mu F$ is an initial $\barF$-algebra with
  the structure $J\ini = \eta_{\mu F} \cdot \ini\colon F(\mu F)
  \to T(\mu F)$. 
\end{proposition}

Coalgebraic trace semantics can be defined when the extended initial
algebra above is also a terminal coalgebra for $\barF$.
%
\begin{theorem}\label{T:Kl}
  Let $F$ be a functor and $T$ a monad on the category
  $\C$. Assume that $\Kl T$ is left strictly $\DCPOb$-enriched and that $F$ has a locally monotone extension
  $\barF$ on $\Kl T$ and an
  initial algebra $(\mu F, \ini)$. Then $(\mu F, J\ini^{-1})$ is a terminal
  coalgebra for $\barF$.
\end{theorem}
\begin{proof}
  Immediate from \autoref{P:ini} and \autoref{T:dcpo}.
  \qed
\end{proof}
\noindent
Compared to the previous result for $\Set$~\cite[Thm.~3.3]{HasuoEA07}
our assumption on the enrichment of the Kleisli category is slightly
stronger; in op.~cit.~only enrichment in $\omega$-cpos is required. A
related result~\cite[Thm.~5.3.4]{Jacobs16} for general
base categories uses enrichment in directed-complete partial
orders. However, in contrast to both of these results, we do not
require that $\Kl T$ has a zero object and, most notably, we only need
the mere existence of $\mu F$ and not that the initial algebra for $F$
is obtained by the first $\omega$ steps of the initial-algebra chain,
that is, as the colimit of the $\omega$-chain given by $F^n 0$
($n < \omega$). The technical reason for this is that the proof of
\autoref{T:dcpo} does not make use of the classical limit-colimit
coincidence technique used e.g.~by Smyth and Plotkin in their seminal
work~\cite{SmythPlotkin82}. Consequently, our proof is easier and
shorter than the previous ones.
\begin{definition}[Coalgebraic Trace Semantics]\label{D:tr}
  Given $F$ and $T$ on $\C$ satisfying the assumptions in
  \autoref{T:Kl} and a coalgebra $c\colon X \to TFX$. The
  \emph{coalgebraic trace map} is the unique coalgebra homomorphism
  $\tr_c$ from $(X, c)$ to $(\mu F, J\ini^{-1})$; that is, the
  following diagram commutes in $\Kl T$:
  \begin{equation}\label{eq:trc}
    \begin{tikzcd}[column sep = 30, row sep = 15]
      X \ar{d}[swap]{c} \ar{r}{\tr_c}
      &
      \mu F
      \ar{d}{J\ini^{-1}}
      \\
      \barF X
      \ar{r}{\bar F \tr_c}
      &
      \barF(\mu F)
    \end{tikzcd}
  \end{equation}
\end{definition}
Among the instances of coalgebraic trace semantics are the trace
semantics of labelled transition systems with explicit
termination~\cite{HasuoEA07}, which are the coalgebras for the set
functor $\pow(1 +\Sigma \times X)$ and that of probabilistic
labelled transitions systems~\cite[Ch.~4]{Hasuo08}, which are the
coalgebras for the set functor $\mathcal D_\leq(1 + \Sigma \times X)$,
where $\mathcal D_\leq$ denotes the subdistribution monad.

\subsection{Coalgebraic Trace Semantics of Non-deterministic Nominal
  Systems}

We will now work towards showing that the semantics of nominal
automata is an instance of the coalgebraic trace semantics. To this
end we will instantiate \autoref{T:Kl} to $\C=\Nom$,
$FX=1+\names\times X$ and $T=\powfs$ (for NOFAs), or to
$FX = 1 + \names \times X + [\names] X$ and $T = \powufs$ (for RNNAs),
\newcommand{\remautorefname}{Remarks}%
cf.~\autoref{rem:nofa-as-coalgebras} and~\ref{rem:rnna-as-coalgebras}.
\renewcommand{\remautorefname}{Remark}%
More generally, we show that every endofunctor arising from a nominal
algebraic signature in the sense of Pitts~\cite[Def.~8.2]{Pitts13} has
a locally monotone extension to $\Kl\powufs$. For $T = \powfs$ most
of the development works out, as we shall see. However, the
distributive law for the abstraction functor in the proof of
\autoref{P:abs-dist} is not well-defined for $\powfs$.

But the first obstacle is that the nominal sets $\powfs X$ and
$\powufs X$ are in general no complete lattices (and not even
$\omega$-cpos) since the union of a chain of (uniformly) finitely
supported sets may fail to be (uniformly) finitely supported.%
\smnote{TODO: We need to mention a concrete counterexampel (in the
  appendix)!}  In this light, the following result is slightly
surprising.

\begin{proposition}\label{P:dcpo}
  For every pair $X,Y$ of nominal sets, the sets $\Kl\powfs(X,Y)$ and $\Kl\powufs(X,Y)$
  form complete lattices (whence dcpos with bottom).
\end{proposition}
\begin{corollary}\label{C:ini-ter}
  If a locally monotone endofunctor $H$ on $\Kl\powfs$ or $\Kl\powufs$ has an initial
  algebra $(\mu H, \ini)$, then $(\mu H, \ini^{-1})$ is its terminal
  coalgebra.
\end{corollary}
\noindent
This is a consequence of~\autoref{T:DCPOb} since the composition in
$\Kl\powfs$ and $\Kl\powufs$ is easily seen to preserve the bottom
(empty set)
on the left and all joins (unions).

\mysubsec{Extending functors to $\Kl\powfs$ and $\Kl\powufs$.} We now show that
endofunctors arising from a nominal algebraic signature (with one name
and one data sort)~\cite[Def.~8.2]{Pitts13} have a canonical locally
monotone extension to $\Kl\powufs$. For instance, the functor~$F$ used
for RNNAs has a locally monotone extension $\barF$ on $\Kl\powufs$.
\begin{definition}
  The class of \emph{binding polynomial functors} on $\Nom$ is the smallest
  class of functors containing the constant and identity and abstraction
  functors and being closed under coproducts, finite products and composition. 
\end{definition}
\noindent In other words, binding polynomial functors are formed
according to the grammar:
\begin{equation}\label{eq:grammar}
  \textstyle
  F ::= C \mid \Id \mid [\names](-) \mid F \times F \mid \coprod_{i\in
    I} F_i \mid FF,
\end{equation}
where $C$ ranges over all constant functors on $\Nom$ and $I$ is an
arbitrary index set. Functors arising from a binding signature in the
sense of Fiore et al.~\cite{FioreEA99} and those associated to a nominal
algebraic signature with one name sort and one data sort (see
Pitts~\cite[Def.~8.12]{Pitts13}) are instances of binding
polynomial functors.
\begin{proposition}\label{P:comm}\smnote{Aren't there any results known in the literature that the
    power object monad on a topos (satisfying s.th.) is commutative?}%
  The monads $\powfs$ and $\powufs$ are commutative w.r.t.~to the
  strengths obtained by restricting the one in~\eqref{eq:pstr}.
\end{proposition}
\begin{proposition}\label{P:abs-dist}
  The abstraction functor $[\names](-)$ has a locally monotone
  extension on $\Kl \powufs$.
\end{proposition}
\begin{proof}[Sketch]
  One uses \autoref{R:Kleisli}\ref{R:Kleisli:3}: the abstraction
  functor is a quotient of the functor $FX = \names \times X$ which is
  equipped with the canonical distributive law
  $\lambda_X\colon \names \times \powufs X \to \powufs(\names \times
  X)$ obtained using the strength of $\powufs$
  (\autoref{E:ext}\ref{E:ext:3} and cf.~\eqref{eq:pstr}). The maps
  $\rho_X\colon [\names](\powufs X) \to \powufs([\names] X)$ are
  defined by $\rho_X (\braket a S) = \set{\braket a s : s \in S}$.
  \qed
\end{proof}
\begin{rem}
  \smnote{Das ist Ex.~4.5 aus Üsames Projektarbeit.}%
  For the monad $\powfs$ our proof does not work. The problem is
  that~$\rho_X$ above is not well-defined in general if $S$ is
  not uniformly finitely supported. For example, for
  $\names \in \powfs \names$ we have
  $\braket a \names = \braket b \names$ for every pair $a, b$ of
  names. However, if $a \neq b$, then the sets
  $\set{\braket a c : c \in \names}$ and
  $\set{\braket b c : c \in \names}$ differ: $\braket a b$ is
  contained in the former but not in the latter set. In fact, since $a \neq b$,
  $\braket a b = \braket b c$ can hold only if $a\fresh\set{b,c}$ and
  $b = (a\, b)\cdot c$ (see Pitts~\cite[Lem.~4.3]{Pitts13}). The
  latter means that $c = a$ contradicting freshness of $a$. 
\end{rem}
\begin{corollary}\label{cor:extension-pufs}
  Every binding polynomial functor has a canonical locally monotone extension to
  $\Kl\powufs$. 
\end{corollary}

Unsurprisingly, an analogous result holds for polynomial functors and
$\powfs$ by the same reasoning applied to a grammar as in~\eqref{eq:grammar}
that does not include the abstraction functor:
\begin{corollary}\label{cor:extension-pfs}
  Every polynomial functor has a canonical locally monotone extension to
  $\Kl\powfs$. 
\end{corollary}
 
\mysubsec{Nominal Coalgebraic Trace Semantics.} Every binding
polynomial functor~$F$ is finitary and therefore has an initial
algebra. In particular, if $F$ arises from a nominal algebraic
signature, we know from Pitts~\cite[Thm.~8.15]{Pitts13} its initial
algebra $\mu F$ is carried by the nominal set of terms modulo
$\alpha$-equivalence (defined in Def.~8.6 of op.~cit.) of the nominal
algebraic signature. If $F$ is polynomial, then $\alpha$-equivalence
is trivial and $\mu F$ the usual set of terms. By \autoref{C:ini-ter}
we have
\begin{corollary}
\begin{enumerate}
\item For every polynomial functor $F$ the terminal coalgebra of its
  canonical extension $\barF$ on $\Kl\powfs$ is carried by the nominal
  set $\mu F$.
\item  For every binding polynomial functor $F$ the terminal coalgebra of
  its canonical extension $\barF$ on $\Kl\powufs$ is carried by the
  nominal set $\mu F$.
\end{enumerate}
\end{corollary}
\noindent
According to \autoref{D:tr} we can thus define a coalgebraic trace
semantics for every coalgebra $X \to \powfs FX$ with $F$ a polynomial functor, as well as for every coalgebra $X\to \powufs FX$ with $F$ a binding
polynomial functor. We now instantiate this to the two types of nominal automata introduced in 
\autoref{sec:nom-aut}.

\mysubsec{Coalgebraic Trace Semantics of NOFAs.}  Recall from
\autoref{rem:nofa-as-coalgebras} that NOFAs are coalgebras
$X\to \powfs F X$ where $FX=1+\names\times X$ on $\Nom$.

\begin{proposition}\label{prop:ini-nofa}
  The initial algebra for $F$ is the nominal set $\names^*$ with
  structure $\iota\colon 1+\names\times \names^* \to \names^*$ defined
  by $\iota(\ast)=\epsilon$ and $\iota(a,w)=aw$.
\end{proposition}
Indeed, the functor $F$ arises from from the algebraic signature with
a constant $\epsilon$ and unary operations $a(-)$ for every
$a \in \names$, and clearly the corresponding term algebra is
isomorphic to the algebra $\names^*$.
\begin{corollary}\label{C:ini}
  The terminal coalgebra for the extension
  $\barF\colon \Kl\powfs \to \Kl\powfs$ is $(\names^*, J\ini^{-1})$
  for $\ini$ from~\autoref{prop:ini-nofa}.
\end{corollary}

\begin{theorem}\label{T:NOFA-Kl}
  For every NOFA $c\colon X \to \powfs FX$ its coalgebraic trace map
  $\tr_c\colon X \to \powfs(\names^*)$ assigns to every state of $X$
  its accepted data language.
\end{theorem}

\noindent
Indeed, one readily works out that assigning to every state of $X$ its
data language is a coalgebra homomorphism from $(X, c)$ to $(\mu F,
J\iota^{-1})$ in $\Kl\powfs$. 

\mysubsec{Coalgebraic Trace Semantics of RNNAs.}  Recall from
\autoref{rem:rnna-as-coalgebras} that RNNAs are coalgebras
$X \to \powufs FX$ where $FX = 1 + \names \times X + [\names] X$ on
$\Nom$.
\begin{proposition}\label{prop:ini}
  The initial algebra for $F\colon \Nom \to \Nom$ is the nominal set
  $\barstr$ of all bar strings modulo $\alpha$-equivalence with the
  algebra structure $\ini\colon 1 +\names \times (\barstr) +
  [\names](\barstr) \to \barstr$ defined by
  \begin{equation}\label{eq:iniF}
    \ini(*)= [\eps]_\alpha,
    \qquad
    \ini(a,[w]_\alpha) = [aw]_\alpha,
    \qquad
    \ini(\braket a [w]_\alpha) = [|aw]_\alpha.
  \end{equation}
\end{proposition}
\noindent
Indeed, the functor $F$ arises from a nominal algebraic signature with
a constant~$\eps$, unary operations $a(-)$ for every $a \in \names$
and one unary name binding operation $\newletter{}$. Terms over this
signature are obviously the same as bar strings. Moreover, it is not
difficult to show that Pitts' notion of $\alpha$-equivalence for
terms~\cite[Def.~8.6]{Pitts13} is equivalent to $\alpha$-equivalence
for bar strings in~\autoref{def:alpha-fin}. Finally, the algebra
structure in~\eqref{eq:iniF} above corresponds to the one given by term formation by
Pitts~\cite[Thm.~8.15]{Pitts13}. Using \autoref{T:Kl} we thus
obtain the following result.
\begin{corollary}
  The terminal coalgebra for the extension
  $\barF\colon \Kl\powufs \to \Kl\powufs$ is $(\barstr, J\ini^{-1})$
  for $\ini$ from~\eqref{eq:iniF}.\smnote{The next overfull goes away
    if we switch the document from `draft' to `final'!}%
\end{corollary}
\begin{theorem}\label{T:RNNA-Kl}
  For every RNNA $c\colon X \to \powufs FX$ its coalgebraic trace map
  $\tr_c\colon X \to \powufs(\barstr)$ assigns to every state of $X$
  its accepted bar language.
\end{theorem}
Indeed, one readily works out that assigning to every state of $X$ its
bar language is a coalgebra homomorphism from $(X, c)$ to $(\mu F,
J\iota^{-1})$ in $\Kl\powufs$. 

\section{Coalgebraic Language Semantics}
\label{S:lang}

In this section we shall see that the language semantics of NOFAs
and RNNAs is an instance of coalgebraic language
semantics~\cite{bms13}. The latter is based on the generalized
determinization construction by Silva et al.~\cite{SilvaEA13}. Here
one considers coalgebras for a functor $GT$, where $T$ models a
branching type and $G$ models the type of transition of a system
(similarly as before in the coalgebraic trace semantics, but this time
the order of composition is reversed). Again, we will apply this to
coalgebras in $\Nom$ for functors $GT$, where $T
= \powfs$ and $G$ is functor composed of products and exponentials, or to $T= \powufs$ and $G$ composed of products, exponentials and binding functors. Specifically, we obtain the two desired types
of nominal automata as instances. 

\subsection{A Recap of General Coalgebraic Language Semantics}
\label{S:lan}

We begin by recalling a few fact about liftings of functors to
Eilenberg-Moore categories.
\begin{rem}\label{R:lift}
  Let $G$ be a functor and $(T,\eta,\mu)$ be a monad
  on the category $\C$.
  \begin{enumerate}
  \item The \emph{Eilenberg-Moore} category $\EM T$ consists of
    algebras $(A,a)$ for $T$, that is, pairs formed by an object $A$ and a
    morphism $a\colon TA \to A$ such that $a\cdot \eta_A=\id_A$ and $a \cdot \mu_A = a
    \cdot Ta$. A morphism in $\EM T$ from $(A,a)$ to $(B,b)$ is a
    morphism $h\colon A \to B$ of $\C$ such that $h \cdot a = b \cdot
    Th$. We write $U\colon \EM T \to \C$ for the forgetful functor
    mapping an algebra $(A,a)$ to its underlying object $A$.

  \item A \emph{lifting} of $G$ is an endofunctor $\hatG \colon \EM T \to \EM
    T$ such that $GU = U\hatG$. As shown by Applegate~\cite{Applegate65} (see
    also Johnstone~\cite{Johnstone75}), liftings of $G$ to $\EM T$ are in
    bijective correspondence with distributive laws of $T$ over
    $G$. The latter are natural transformations $\lambda\colon TG \to
    GT$ compatible with the monad structure:
    \[
      G\eta = \lambda \cdot \eta G,
      \qquad
      \lambda \cdot \mu G = G\mu \cdot \lambda T \cdot T\lambda.
    \]

  \item\label{R:lift:3} Suppose that $G$ has a terminal coalgebra
    $(\nu G, \tau)$ and the lifting $\hatG$ on $\EM T$ via the
    distributive law $\lambda$. It follows from the work of Turi and
    Plotkin~\cite{TuriP97} (see also
    Bartels~\cite[Thm.~3.2.3]{Bartels04}) that the terminal coalgebra
    for $G$ lifts to a terminal coalgebra for $\hatG$. In fact, one
    obtains a canonical structure of a $T$-algebra on $\nu G$ by
    taking the unique coalgebra homomorphism $\alpha$ in the diagram
    below:
    \[
      \begin{tikzcd}
        T(\nu G)
        \ar{r}{T\ter}
        \ar{d}[swap]{\alpha}
        &
        TG(\nu G)
        \ar{r}{\lambda_{\nu G}}
        &
        GT(\nu G)
        \ar{d}{G\alpha}
        \\
        \nu G
        \ar{rr}{\ter}
        &&
        G(\nu G)
      \end{tikzcd}      
    \]
    It is then easy to prove that $\alpha$ is indeed the structure of an
    algebra for $T$ and that $\ter\colon \nu G \to G(\nu G)$ is a
    homomorphism of Eilenberg-Moore algebras (in fact, this is
    expressed by the commutativity of the above diagram). Moreover,
    $(\nu G, \ter)$ is the terminal $\hatG$-coalgebra.
  \end{enumerate}
\end{rem}
%

%
\begin{proposition}\label{P:adj-trans}
  Let $T\colon \C \to \C$ be a monad and $L \dashv R\colon \C \to \C$
  an adjunction with the counit $\eps\colon LR \to \Id$. Given a distributive law 
  $\lambda\colon LT \to TL$, we obtain a
  distributive law $\rho\colon TR \to RT$ as the adjoint
  transpose of
  \mbox{\(
    LTR \xra{\lambda R} TLR \xra{T\eps} T.
  \)}
\end{proposition}
\noindent Recall that the \emph{adjoint transpose} of a morphism $LX
\to Y$ is the corresponding morphism $X \to RY$ under the natural
isomorphism $\C(LX, Y) \cong \C(X,RY)$.
\begin{expl}\label{E:lift}
  \begin{enumerate}
  \item The identity functor on $\C$ obviously
    lifts to $\EM T$, and so does a constant functor on the carrier object
    of an Eilenberg-Moore algebra for~$T$. 

    
  \item Suppose that $\C$ has products. Then for a product $F \times
    G$ of functors one uses that $U$ preserves products. Given
    liftings $\hatF$ and $\hatG$, it is clear that $\hatF \times
    \hatG$ is a lifting of $F\times G$. This works similarly for
    arbitrary products.

  \item Suppose that $\C$ is cartesian closed and that the monad $T$
    is strong (cf.~\autoref{E:ext}\ref{E:ext:3}). Then the
    exponentiation functor $(-)^A$ lifts to $\EM T$ for every object
    $A$ of $\C$. In fact, we apply \autoref{P:adj-trans} to the
    adjunction $A \times (-) \dashv (-)^A$ and use that two of the
    axioms of the strength $A \times TX \to T(A \times X)$ state that
    it is a distributive law of $A\times (-)$ over $T$.
    \takeout{
    it is an easy exercise
    (see~\cite[Exercise~5.2.16]{Jacobs16}) to prove that one obtains a
    distributive law $\lambda$ with the component at an object $X$
    given by currying the following morphisms
    \[
      Y \times T(X^Y) \xra{s_{Y,X^Y}} T(Y \times X^Y) \xra{T\ev} TX,
    \]
    where $s$ is the strength of $T$ and $\ev\colon Y \times X^\Y \to
    X$ is the evaluation morphism (a component of the counit of the
    adjunction $Y \times (-) \dashv (-)^Y$.}
  \end{enumerate}
\end{expl}
\begin{rem}\label{R:hom-ext}
  Recall that, for every monad $(T,\eta,\mu)$ on $\C$, the pair
  $(TX,\mu_X)$ is the free algebra for $T$ on $X$ with the
  universal morphism $\eta_X\colon X \to TX$. Given an Eilenberg-Moore
  algebra $(A,a)$ for $T$ and a morphism $f\colon X \to A$ in
  $\C$, we have a unique morphism  $f^\sharp\colon (TX,\mu_X) \to
  (A,a)$ in $\EM T$ such that $f^\sharp \cdot \eta_X
  = f$. We call $f^\sharp$ the \emph{homomorphic extension} of $f$.
\end{rem}
\begin{construction}[Generalized Determinization~\cite{SilvaEA13}]\label{C:gen-det}
  Let $T$ be a monad on the category $\C$ and $G$ an endofunctor on
  $\C$ having a lifting $\hatG$ on $\EM T$. Given a coalgebra
  $c\colon X \to GTX$ its \emph{(generalized) determinization} is the
  $G$-coalgebra obtained by taking the homomorphic extension
  $c^\sharp\colon TX \to GTX$ using that $\hatG(TX, \mu_X)$ is an
  algebra for $T$ carried by $GTX$.
\end{construction}
\takeout{
\begin{rem}\label{R:gen-det}
  It is straightforward to establish that for every homomorphism
  $h\colon (X,c) \to (Y,d)$ we obtain that $Th\colon (TX,c^\sharp) \to
  (TY,d^\sharp)$ is a homomorphism of coalgebras for $\hatG$. It
  follows that generalized determinization
  is the object assigment of a functor $D\colon \Coalg TF \to \Coalg
  \hatG$.  
\end{rem}}
\noindent
Among the instances of this construction are the well-known power-set construction of
deterministic automata~\cite{SilvaEA13} as well as the non-determinization of alternating
automata and that of Simple Segala systems~\cite{JacobsEA15}.

\begin{defn}[Coalgebraic Language Semantics~\cite{bms13}]\label{D:coalg-lan-sem}
  Given $T$, $G$ and a coalgebra $c\colon X \to GTX$ as in \autoref{C:gen-det}, the
  \emph{coalgebraic language morphism} $\lan c\colon X \to \nu G$ is the
  composite of the unique coalgebra homomorphism $h$ from the
  determinization of $(X,c)$ to $\nu G$ with the unit $\eta_X$ of the
  monad $T$, which is summarized in the diagram on the left below:
\end{defn}

\vspace*{-5pt}
\noindent
\begin{minipage}[t]{4.75cm}
  \hspace*{8pt}
  \(
    \begin{tikzcd}[column sep = 10]
      X
      \ar{d}[swap]{c}
      \ar{r}{\eta_X}
      \ar[shiftarr = {yshift=15}]{rr}{\lan c}
      &
      TX
      \ar{r}{h}
      \ar{dl}[swap]{c^\sharp}
      &
      \nu G
      \ar{d}{\ter}
      \\
      GTX
      \ar{rr}{Gh}
      &
      &
      G(\nu G)
    \end{tikzcd}
  \)
\end{minipage}
\hfill
\begin{minipage}{7.3cm}
\noindent
Among the instances of coalgebraic language semantics are, of course,
the language semantics of non-deterministic~\cite{SilvaEA13,JacobsEA15}, weighted and
probabilistic automata, but also the languages generated by
context-free grammars~\cite{WinterEA13,mpw20}, constructively
$\dsS$-algebraic formal power series for  a semiring $\dsS$ 
\end{minipage}

\noindent
(the `context-free' weighted languages)~\cite{WinterEA15,mpw20}. Less
direct instances are the languages accepted by machines with extra
memory such as (deterministic) push-down automata and Turing
machines~\cite{gms20}.%

\mysubsec{Relation of Coalgebraic Trace and Language Semantics.}
Jacobs et al.~\cite{JacobsEA15} show how the coalgebraic trace
semantics and coalgebraic language semantics are connected in cases
where both are applicable. We give a terse review of this including a
proof (see appendix) of the result of op.~cit.~that we use here.
\begin{assumption}
  We assume that $T$ is a monad and $F, G$ are endofunctors, all on
  the category $\C$, such that $F$ has the extension $\barF$ on
  $\Kl T$ via the distributive law $\lambda\colon FT \to TF$ and $G$
  has the lifting $\hatG$ on $\EM T$ via the distributive law
  $\rho\colon TG \to GT$. Moreover, we assume that we have an
  \emph{extension} natural transformation $\eps\colon TF \to GT$
  compatible with the two distributive laws:
  \begin{equation}\label{diag:eps}
    \begin{tikzcd}
      TFT
      \ar{r}{T\lambda}
      \ar{d}[swap]{\eps T}
      &
      TTF
      \ar{r}{\mu F}
      &
      TF
      \ar{d}{\eps}
      \\
      GTT
      \ar{rr}{G\mu}
      &&
      GT
    \end{tikzcd}
    \qquad\quad
    \begin{tikzcd}
      TTF
      \ar{rr}{\mu F}
      \ar{d}[swap]{T\eps}
      &&
      TF
      \ar{d}{\eps}
      \\
      TGT
      \ar{r}{\rho T}
      &
      GTT
      \ar{r}{G \mu}
      &
      GT
    \end{tikzcd}
  \end{equation}
\end{assumption}
\begin{rem}\label{R:eps}
  \begin{enumerate}
  \item\label{R:eps:1} For every object $X$ of $\C$ the morphism $\eps_X$ is
    a homomorphism of Eilenberg-Moore algebras for $T$ from
    $(TFX, \mu_{FX})$ to $\hatG(TX, \mu_X)$. Indeed, this is precisely
    what the commutativity of the diagram on the right
    in~\eqref{diag:eps} expresses.
    
  \item\label{R:eps:2} For every coalgebra $c\colon X \to TFX$ the
    extension natural transformation yields a coalgebra
    $\eps_X \cdot c\colon X \to GTX$, and we take its determinization
    $(TX, (\eps_X \cdot c)^\sharp)$. This is the object assignment of
    the functor $E\colon \Coalg \barF \to \Coalg \hatG$ which maps an
    $\barF$-coalgebra homomorphism $h\colon (X,c) \to (Y,d)$ to
    $Eh = h^\sharp\colon TX \to TY$, the homomorphic extension of
    $h\colon X \to TY$ (in $\C$). One readily proves that $h^\sharp$
    is a $\hatG$-coalgebra homomorphism using the naturality of $\eps$
    as well as the laws in~\eqref{diag:eps}. Functoriality follows
    since $E$ is clearly a lifting of the canonical comparison functor
    $\Kl T \to \EM T$; see Jacobs et al.~\cite[Thm.~2]{JacobsEA15} for
    the proof, and we include a proof in the appendix for the
    convenience of the reader.
    
  \item\label{R:eps:3} We obtain a canonical morphism
    $e\colon T(\mu F) \to \nu G$ by applying the functor $E$ to the
    coalgebra $J\ini^{-1}\colon \mu F \to TF(\mu F)$
    (cf.~\autoref{P:ini}) and taking the unique coalgebra homomorphism
    from it to the terminal $\hatG$-coalgebra
    (\autoref{R:lift}\ref{R:lift:3}).
  \end{enumerate}
\end{rem}

Now recall the coalgebraic trace semantics from \autoref{D:tr}. The
following result follows from Jacobs et al.'s
result~\cite[Prop.~5]{JacobsEA15}.
\begin{proposition}\label{P:relation}
  For every coalgebra $c\colon X \to TFX$ we have
  \[
    \lan(\eps_X\cdot c) = \big(X \xra{\tr_c} T(\mu F) \xra{e} \nu G \big).
  \]
\end{proposition}
\subsection{Coalgebraic Language Semantics of Nominal Systems}

We will now work towards that the language semantics of nominal
automata is an instance of coalgebraic language semantics. To this end
we will instantiate the results of \autoref{S:lan} to $\C= \Nom$,
$GX = 2 \times X^\names$ and $T = \powfs$ (for NOFAs), or to
$GX = 2 \times X^\names \times [\names]X$ and $T= \powufs$ (for
RNNAs). More generally, in the former case we show that certain
polynomial functors $G$ with exponentiation lift to $\EM\powfs$, and in the latter
case, certain binding polynomial functors with exponentation lift to
$\EM\powufs$. For our specific instances of interest we show that
the terminal coalgebra $\nu G$ is given by (data or bar)
languages. The desired end result then follows by an application of
\autoref{P:relation}.

The class of functors $G$ we consider are formed according
to the grammar
\begin{equation}\label{eq:grammar-2}
  \textstyle
  G ::= A \mid \Id \mid [\names](-) \mid \prod_{i \in I} G_i \mid G^N,
\end{equation}
where $A$ ranges over all nominal sets equipped with the structure
$a\colon \powufs A \to A$ of an algebra for the monad $\powufs$, $I$
is an arbitrary index set, and $N$ ranges over all nominal sets.
Every such functor $G$ has a canonical lifting to $\EM \powufs$.
This can be proved by induction over the grammar using
\autoref{E:lift} and 
\begin{proposition}\label{P:lift-abst}
  The abstraction functor has a canonical lifting to $\EM \powufs$.
\end{proposition}
\begin{proof}
  The abstraction functor $[\names](-)$ has a left-adjoint $\names *
  (-)$, where $*$ denotes the \emph{fresh product} defined for two
  nominal sets $X$ and $Y$ by
  \[
    X * Y = \set{(x,y) : x \in X,\, y\in Y,\, \supp(x) \cap\supp(y) = \emptyset},
  \]
  see \cite[Thm.~4.12]{Pitts13}. The strength of $\powufs$ restricts to the fresh product; we
  have
  \[
    s_{X,Y}\colon X * \powufs Y \to \powufs(X * Y)
    \qquad
    (x,S) \mapsto \set{(x,s) : s \in S}.
  \]
  Indeed, if $\supp(x) \cap \supp(S) = \emptyset$, then
  $\supp (x) \cap \supp (s) = \emptyset$ for every $s\in S$ because
  $S$ is uniformly finitely supported and thus
  $\supp(s) \subseteq \supp(S)$. It follows that
  $s_{\names, X}\colon \names * \powufs X \to \powufs (\names * X)$
  yields a distributive law of $\names *(-)$ over $\powufs$. By
  \autoref{P:adj-trans} we thus obtain a distributive law of $\powufs$
  over $[\names] (-)$.  \qed
\end{proof}
\begin{corollary}\label{C:lift-nu}
  For every functor $G$ according to the grammar in
  \eqref{eq:grammar-2} the terminal coalgebra $\nu G$ lifts to a
  terminal coalgebra of $\hatG$ on $\EM\powufs$.
\end{corollary}
\noindent
The terminal coalgebra $\nu G$ exists since every such $G$ is an
accessible functor on $\Nom$. This can be shown by induction on the
structure of $G$; for exponentiation in the induction step one argues
similarly as Wißmann~\cite[Cor.~3.7.4]{Wissmann20} has done for
orbit-finite sets: an exponentiation functor $(-)^N$ is
$\lambda$-accessible iff the set of orbits of $N$ has cardinality less than $\lambda$.
\smnote{Or is there a better reference for this?}
Now use \autoref{R:lift}\ref{R:lift:3}.%

Consequently, one can define a coalgebraic language semantics for
every functor $G$ according to the grammar~\eqref{eq:grammar-2}.

\begin{rem}
  \begin{enumerate}
  \item For $T = \powfs$ one has the same results for functors $G$ on
    $\Nom$ according to the reduced grammar obtained from the one
    in~\eqref{eq:grammar-2} by dropping the abstraction functor
    $[\names](-)$. In fact, a functor according to the reduced grammar
    has a canonical lifting to $\EM T$ whenever $T$ is a strong monad
    on a cartesian closed category (by \autoref{E:lift}).

  \item We have dropped the abstraction functor in the previous item
    because our proof of \autoref{P:lift-abst} does not work for
    $\powfs$. The problem is that the strength in~\eqref{eq:pstr} does
    not restrict to the fresh product for all finitely supported
    subsets. Indeed, even if $\supp(x)$ and $\supp(S)$ are disjoint,
    the support of $x$ may not be disjoint from that of every element
    $s\in S$, whence $(x,s)$ does not lie in $X * Y$. For example,
    take $X = Y = \names$ and $S = \names \setminus \set {a}$ for some
    $a \in \names$. Clearly, $\supp(S) = \set{a}$. Thus, for every
    $b \neq a$, we see that $(b,S)$ lies in
    $\names * \powfs\names$. However, while $b \in S$ we do not
    have that
    $(b,b) \in \names * \names = \set{(a,a') : a,a' \in \names,
      a\neq a'}$, which means that $s_{\names,\names}(b, S)$ does not
    lie in $\powfs(\names * \names)$.
    \takeout{ 
    %
    For the monad $\powfs$ the above proof does not work. The problem
	  is that the isomorphism $\psi_X$ is not a distributive law because
	  the diagram
	  \[
			\begin{tikzcd}
			\powfs\powfs({[\names]}X) \ar{rr}{\mu_{[\names]X}}
			\ar{d}[swap]{\powfs\psi_X}
			& &
			\powfs({[\names]}X) \ar{d}{\psi_X}
			\\
			\powfs({[\names]}(\powfs X)) \ar{r}{\psi_{\powfs X}}
			&
			{[\names]}(\powfs\powfs X) \ar{r}{[\names]\mu_X}
			&
			{[\names]}\powfs X
			\end{tikzcd}
	 \]
	does not commute. Indeed, fix $a \in \names$, then $\mathscr{S} = \set{\set{\braket{a}b}: 
	b \in \names} \in \powfs\powfs({[\names]}\names)$ with support $\set{a}$,
	while every element of $\mathscr{S}$ has support $\set{b} \setminus \set{a}$.
	If the diagram above commuted, the equality
	\[ 
		\braket{c}\set{x: \braket{c}x \in S \text{ for some } S \in \mathscr{S}}
		=
		\braket{d}\set{x: \braket{e}x \in S \text{ for some } S \in \mathscr{S},\ e \fresh S}
	\]
	would need to hold, for any $c$ that is fresh for $\set{\braket{a}b: b \in \names}$
	and $d$ that is fresh for $\powfs\psi_X\mathscr{S}$. Because of Pitts'
	Choose-a-Fresh-Name-Principle~\cite[p.~49]{Pitts13}, we can assume $d = c$,
	and $d \neq a$, because $a$ supports $\set{\braket{a}b: b \in \names}$.
	But then the equality cannot hold, because the set on left side does not
	contain $a \in \names$: Since $d \neq a$, $\braket{d}a = \braket{a}b$ can
	hold only if $d \fresh \set{a,b}$ and $a = (d\,a)\cdot b$ (see
	Pitts~\cite[Lem.~4.3]{Pitts13}). The latter means that $b = d$ contradicting
	freshness of $d$ for it. However, the set of the right side does contain
	$a \in \names$, since $a \fresh \set{\braket{a}a}$.}
  \end{enumerate}
\end{rem}

\mysubsec{Coalgebraic Language Semantics of NOFAs} We now apply the
previous results to $T = \powfs$ and $GX = 2 \times X^\names$.
\begin{rem}\label{R:exp}
  We have a canonical isomorphism $\powfs(\names
  \times X) \cong (\powfs X)^\names$ given by $S \mapsto (a \mapsto
  \set{x: (a,x) \in S})$. This follows from the fact that
  $\powfs$ is the power object functor on the topos $\Nom$ and so we
  have $\powfs X \cong 2^X$.
\end{rem}
Consequently, a NOFA may be regarded as a coalgebra for $G\powfs$:
\[
  X\to \powfs(1+\names\times X)\cong 2\times (\powfs X)^\names =
  G\powfs X.
\]

\begin{proposition}\label{prop:ter-nofa}
  The terminal coalgebra for $G$ is the nominal set $\datalang$ of all
  data languages with the structure
  \[
    \datalang \xra{\ter} 2 \times \datalang^\names,\quad
    L \mapsto (b, a \mapsto a^{-1}L),
  \]
  where $b = 1$ if
  $\epsilon\in L$ and $0$ else, and $a^{-1}L=\{w\in \names^* : aw\in L \}$.
\end{proposition}
The proof is analogous to the one that for every alphabet $A$ the set
functor $X\to 2\times X^A$ has the terminal coalgebra $\pow (A^*)$,
see e.g.~Rutten~\cite{rutten00}.

We may thus define the coalgebraic language semantics for NOFAs as in
\autoref{D:coalg-lan-sem}.
\begin{rem}
  We take $FX = 1 + \names \times X $ as in
  \autoref{T:NOFA-Kl} and obtain $\mu F = \names^*$
  (\autoref{prop:ini-nofa}) and $\nu G = \datalang$
  (\autoref{prop:ter-nofa}).
  Moreover, analogous to ordinary non-deterministic
  automata~\cite[Sec.~7.1]{JacobsEA15}, we have an extension natural
  transformation
  $\eps_X\colon \powfs(1 + \names \times X) \to 2\times (\powfs
  X)^\names$ given by
  \[
    \eps_X(S) = (b, a \mapsto S_a),
  \]
  where $b = 1$ iff the element $*$ of $1$ lies $S$ and
  $S_a = \set{x : (a,x) \in S}$. The ensuing canonical morphism
  $e\colon \powfs(\mu F) \to \nu G$ from \autoref{R:eps}\ref{R:eps:3}
  is then easily seen to be just the identity map on $\datalang$.
\end{rem}
\begin{corollary}\label{C:NOFA-EM}
  The coalgebraic language semantics assigns to each state of a NOFA
  the data language it accepts.
\end{corollary}
\noindent
Indeed, this follows from \autoref{T:NOFA-Kl} and \autoref{P:relation}
using that in the latter result
$e$ is the identity map on $\datalang$.

\mysubsec{Coalgebraic Language Semantics of RNNAs} We now apply the
previous results to $T = \powufs$ and $GX = 2 \times X^\names \times
[\names]X$.
\begin{rem}
  \begin{enumerate}
  \item The canonical isomorphism from \autoref{R:exp}
    restricts to an injection
    $i\colon \powufs(\names \times X) \monoto (\powufs
    X)^\names$. Indeed, take a uniformly finitely supported subset
    $S \subseteq \names \times X$. Then for every $a \in \names$, every
    element $x$ of the set $i(S)(a) = \set{x : (a,x) \in S}$ satisfies
    $\supp(x) \subseteq \set{a} \cup \supp (x) = \supp (a,x) \subseteq \supp(S)$
    and therefore that set lies in $\powufs X$. However, note that the
    inverse of the isomorphism from \autoref{R:exp} does not restrict to
    uniformly finitely supported subsets.

  \item The components
    $\rho_X\colon [\names]\powufs X \to \powufs([\names] X)$ of the distributive law from the
    proof of \autoref{P:abs-dist} are in fact isomorphisms with
    inverses $\psi_X\colon \powufs([\names] X) \to [\names]\powufs X$
    defined by $\psi_X(S) = \braket a \set{x : \braket a x \in S}$,
    where $a$ is fresh for $S$. These inverses can also be gleaned
    from Pitts' result~\cite[Prop.~4.14]{Pitts13} which shows that the
    abstraction functor preserves exponentials specializing to
    $\powfs([\names] X) \cong[\names]\powfs X$. However, note that
    $\rho_X$ has a more involved description in the case
    of $\powfs$.
  \end{enumerate}
\end{rem}
It follows that for every nominal set $X$ we have an injection
\begin{equation}\label{eq:inj}
  m_X\colon 2 \times \powufs(\names \times X) \times
  \powufs([\names] X)
  \monoto
  2 \times (\powufs X)^\names \times [\names] (\powufs X).
\end{equation}
Thus every RNNA (\autoref{rem:rnna-as-coalgebras}) may be regarded as
a coalgebra for $G\powufs$.

A description of the terminal coalgebra for $G$ has previously been
given by Kozen et al.~\cite[Thm.~4.10]{KozenEA15}. We provide a
different (of course, isomorphic) description as a final ingredient
for our desired result.
\begin{proposition}\label{prop:ter}
  The terminal coalgebra for $G$ is the nominal set $\barlang$ of all
  bar languages with the structure
  \[
    \barlang \xra{\ter} 2 \times (\barlang)^\names \times
    [\names]\barlang,\quad
    S \mapsto (b, a \mapsto S_a, S_{\scriptnew a}),
  \]
  where $b = 1$ if
  $[\eps]_\alpha \in S$ and $0$ else,
  $S_a = \set{[w]_\alpha : [aw]_\alpha \in S}$ and
  $S_{\scriptnew a} = \braket a \set{[w]_\alpha : [\newletter a
    w]_\alpha \in S}$ for any $a$ which is fresh for $S$.%
  \smnote{Is the latter set really correct like this?}
\end{proposition}

We may thus define the coalgebraic language semantics for RNNAs as in
\autoref{D:coalg-lan-sem}.
\begin{rem}\label{R:epsX}
  We take $FX = 1 + \names \times X + [\names] X$ as in
  \autoref{T:RNNA-Kl} and obtain $\mu F = \barstr$
  (\autoref{prop:ini}) and $\nu G = \barlang$
  (\autoref{prop:ter}). We also define a natural
  transformation $\eps\colon \powufs F \to G \powufs$ by composing the
  canonical isomorphism $\powufs(1 +  \names \times X + [\names] X)
  \cong 2 \times \powufs(\names \times X) \times \powufs([\names] X)$
  with the injection $m_X$ from~\eqref{eq:inj}. 
  For every uniformly finitely supported subset $S \subseteq 1 +
  \names \times X + [\names] X$ we have
    \(
    \eps_X(S) = (b, a \mapsto S_a, S_{\scriptnew a}),
    \)
  where $b = 1$ iff the element $*$ of $1$ lies in $S$, $S_a = \set{s :
    (a,s) \in S}$ and $S_{\scriptnew a} = \braket a \set{s : \braket a
    s \in S}$, where $a$ is fresh for (all elements $\braket b s$ in)
  $S$. 
\end{rem}
\begin{lemma}\label{L:eps-ext}
  The natural transformation $\eps\colon \powufs F \to G \powufs$ is
  an extension.
\end{lemma}
\begin{lemma}\label{L:incl}
  The canonical morphism $e\colon \powufs(\mu F) \to \nu G$ from
  \autoref{R:eps}\ref{R:eps:3} is the inclusion map
  $\powufs(\barstr) \subto \barlang$.
\end{lemma}
\begin{corollary}\label{C:RNNA-EM}
  The coalgebraic language semantics assigns to each state of an RNNA
  the bar language it accepts.
\end{corollary}
\noindent
Indeed, this follows from \autoref{T:RNNA-Kl} and \autoref{P:relation}
using that in the latter result
$e\colon\powufs(\barstr) \subto \barlang$ is the inclusion map by \autoref{L:incl}.


\section{Conclusions and Future Work}

We have worked out coalgebraic semantics for two species of
non-deterministic automata for data languages:
NOFAs~\cite{BojanczykEA14} and RNNAs~\cite{SchroderEA17}. We have seen
that their semantics arises both as an instance of the Kleisli style
coalgebraic trace semantics and from the Eilenberg-Moore style
coalgebraic language semantics, which is based on generalized
determinization. To see that both semantics coincide we have employed
the results by Jacobs et al.~\cite{JacobsEA15}.

We have also revisited coalgebraic trace semantics in general and
given a new compact proof of the main extension result
for initial algebras in that theory. Our proof avoids assumptions on
the convergence of the initial algebra chain; mere existence of an
initial algebra suffices.

Having provided coalgebraic semantics for non-deterministic nominal
systems makes the powerful toolbox of coalgebraic methods fully
available to those systems. For example, generic constructions like
coalgebraic $\eps$-elimination~\cite{SilvaW13,bmsz15} can be
instantiated to them. Or coalgebraic up-to techniques
starting with the work by Rot et al.~\cite{RotEA13} might lead to new
proof principles and algorithms, cf.~\cite{BonchiPous13}.

Our general extension and lifting results for nominal systems may be
applied to related kinds of systems, e.g.~nominal transition
systems and the coalgebraic study of equivalences for them. Going a
step beyond the standard coalgebraic trace and language
semantics, graded semantics~\cite{DorschEA19} should lead to a nominal spectrum
of equivalences generalizing van Glabbeek's famous linear time -- branching time
spectrum~\cite{Glabbeek01}. 

\smnote[inline]{Is there other future work we can think of? Didn't we
  discuss something last week?}

%
%
\bibliographystyle{splncs04}
\bibliography{refs}
    
%
%
%
\clearpage
\appendix

\section*{Appendix}

This appendix contains proof details omitted due to space restrictions.

\section{Details for \autoref{S:DCPOb}}

\subsection*{Proof of \autoref{T:dcpo}}

Even though the proof appeared recently we provide full details for
the convenience of the reader and to strengthen our point that the
whole proof of \autoref{T:Kl} is simpler and shorter than the previous ones. 

First, the proof is based on the following fixed point theorem for
directed-complete partial orders.
\begin{theorem}[Pataraia's Theorem]\label{T:Pata}
  Let $P$ be a dcpo with bottom. Then every monotone map $f\colon P
  \to P$ has a least fixed point $\mu f$.
\end{theorem}
\noindent
This result is attributed to Pataraia since he gave the first
constructive proof~\cite{Pataraia97}. Sadly, he never published the
proof in written form. But proofs subsequently appeared in several
sources, e.g.~Ad\'amek et al.~\cite[Thm.~2.4]{amm21} present a proof based to
Martin's presentation~\cite{Martin13}.
%

A shorter but non-constructive argument appears as early as in
Zermelo's 1904 paper~\cite{Zermelo04} proving the well-ordering
theorem. His argument works for a \emph{chain-complete} poset $P$,
which means that for every ordinal $i$ each $i$-chain has a join in
$P$, where an \emph{$i$-chain} is a sequence $(x_j)_{j <i}$ of elements of
$P$ such that $x_j \leq x_k$ for all $j \leq k < i$. By Markowsky's
theorem~\cite{Markowsky76} we know that a poset is chain-complete iff
it is a dcpo with bottom.

\begin{proof}[\autoref{T:Pata}]
  Given a monotone function $f$ on a chain-complete poset $P$, one
  defines an ordinal-indexed sequence $f^i(\bot)$ by the following
  transfinite recursion:
  \[
    f^0(\bot) = \bot,\  f^{j+1}(\bot) = f(f^j(\bot)),
    \ \text{and}\
    f^j(\bot) = \bigvee_{i < j} f^i(\bot)
    \ \mbox{for limit ordinals $j$}.
  \]
  It is easy to verify that this is a chain in $P$. By Hartogs'
  Lemma~\cite{Hartogs15}, there exists an ordinal $i$ such that there
  is no injection from $i$ to the set $P$ (obviously, such an $i$ is
  larger than the cardinality of $P$). Then there must be some
  ordinals $j < k < i$ such that $f^j(\bot) = f^k(\bot)$, which
  implies that $f^{j+1}(\bot) = f^j(\bot)$. So~$f^j(\bot)$ is a fixed
  point of $f$. Now let $j$ be the least ordinal such that $f^j(\bot)$
  is a fixed point, and let $f(x) = x$. An easy transfinite induction
  shows that $f^i(\bot) \leq x$ for all ordinals $i$. Hence, $f^j(\bot)$
  is the least fixed point of~$f$.  \qed
\end{proof}

From the proof we immediately extract an induction principle related
to Scott induction; it appears e.g.~in work by Escard\'o~\cite[Thm.~2.2]{Escardo03} and
Taylor~\cite{Taylor21}.
\begin{corollary}\label{C:Pata}
  Let $P$ be a dcpo with bottom. If $f$ is a monotone function on $P$,
  then~$\mu f$ belongs to every subset $S\subseteq P$ which contains
  $\bot$ and is closed under~$f$ and under directed joins.
\end{corollary}

For the proof of \autoref{T:dcpo} we still need a well-known lemma
establishing \emph{uniformity} of least fixed points. It is readily
proved using the above induction principle. A monotone function $f$ on
a dcpo $D$ with bottom is \emph{continuous} if it preserves directed
joins, and \emph{strict} if $f(\bot) = \bot$.
\begin{lemma}\label{L:mu-pres}
  Let $P, Q$ be dcpos with bottom and let $f\colon P \to P$ and
  $g\colon Q \to Q$ be monotone. For every strict continuous map
  $h\colon P \to Q$ such that $g \cdot h = h \cdot f$ we have $h(\mu
  f) = \mu g$.%
  \smnote{No `and' between $P, Q$ in the first line, please. I'd like
    to keep the mathdisplays within one line.}
\end{lemma}
\begin{proof}
  First, $h(\mu f)$ is a fixed point of $g$: we have
  \(
  g (h(\mu f))
  =
  h(f(\mu f))
  =
  h(\mu f).
  \)
  Therefore $\mu g \leq h(\mu f)$. For the reverse relation, let $S =
  \set{x \in P : h(x) \leq \mu g}$. Since $h$ is strict, we see that
  $\bot \in S$. Moreover, $S$ is closed under $f$, for if $x \in S$ we
  obtain
  \(
  h(f(x)) = g(h(x)) \leq g(\mu g) = \mu g
  \)
  using monotonicity of $g$ in the second step. Finally, $S$ is closed
  under directed joins: if $D \subseteq S$ is a directed set we obtain
  \(
  h(\bigvee D) = \bigvee_{x \in D} h(x) \leq \bigvee_{x\in D} \mu g =
  \mu g,
  \)
  whence $\bigvee D$ lies in $S$. Thus, by \autoref{C:Pata}, $\mu
  f\in S$, which means that $h(\mu f) \leq \mu g$.
  \qed
\end{proof}
\begin{proof}[\autoref{T:dcpo}]
  Let $\ini\colon F I \to I$ be an initial algebra. For every
  coalgebra $\gamma\colon C \to FC$, we prove that a unique
  homomorphism into $(I, \ini^{-1})$ exists.
  \begin{enumerate}
  \item Existence. The endomap $g$ on $\C(C,I)$ given by
    $h \mapsto \ini \cdot Fh \cdot \gamma$ is monotone since
    composition is continuous, whence monotone, and $F$ is locally
    monotone. Hence, it has a least fixed point $h\colon C \to I$ with
    $\ini^{-1}\cdot h = Fh \cdot \gamma$ by \autoref{T:Pata}. This is
    a coalgebra homomorphism.

  \item Uniqueness. First notice that for $\C(I,I)$ we have an the
    analogous endomap $f$ given by $k \mapsto \ini \cdot Fk \cdot
    \ini^{-1}$. Since $I$ is initial, the only fixed point of $f$ is
    $k = \id_I$. Thus $\id_I = \mu f$. Now
    suppose that $h'\colon (C, \gamma) \to (I,\ini^{-1})$ is any
    coalgebra homomorphism. We know that
    $\C(h',I)\colon \C(I,I) \to \C(C,I)$ defined by $k \mapsto k \cdot
    h'$ is a strict continuous map;
    strictness follows from left-strict\-ness of composition:
    $\bot_{I,I} \cdot h' = \bot_{C,I}$. We now show that
    $g \cdot \C(h',I) = \C(h',I) \cdot f$. Indeed, unfolding
    the definitions, we have for every $k\colon I \to I$:
    \begin{align*}
      g \cdot \C(h',I)(k)
      & =
      g(k \cdot h')
      =
      \ini \cdot F(k\cdot h') \cdot \gamma
      =
      \ini \cdot Fk \cdot Fh' \cdot \gamma
      =
      \ini \cdot Fk \cdot \ini^{-1} \cdot h' \\
      & =
      f(k) \cdot h' =  \C(h',I) (f(k)).
    \end{align*}    
    By \autoref{L:mu-pres}, $\C(h',I)(\mu f) = \mu
    g$, which means that $h' = \id_I \cdot h' = h$.\qed%
  \end{enumerate}
\end{proof}

\section{Details for \autoref{S:trace}}

\subsection*{Details for \autoref{R:Kleisli}\ref{R:Kleisli:3}}

The naturality of $\rho$ as well as the two laws of a distributive law
all follow from the corresponding properties of $\lambda$ using that
the components of $q\colon F\epito G$ are epimorphic.

For the naturality of $\rho$ we consider the following diagram for every morphism
$f\colon X \to Y$ of $\C$. 
\[
  \begin{tikzcd}
    FTX
    \ar{rrr}{\lambda_X}
    \ar{ddd}[swap]{FTf}
    \ar[->>]{rd}{q_{TX}}
    &&&
    TFX
    \ar{ddd}{TFf}
    \ar{ld}[swap]{Tq_X}
    \\
    &
    GTX
    \ar{r}{\rho_X}
    \ar{d}[swap]{GTf}
    &
    TGX
    \ar{d}{TGf}
    \\
    &
    GTY
    \ar{r}{\rho_Y}
    &
    TGY
    \\
    TFY
    \ar{rrr}{\lambda_Y}
    \ar[->>]{ru}[swap]{q_{TY}}
    &&&
    FTY
    \ar{lu}{Tq_Y}
  \end{tikzcd}
\]
The outside commutes by the naturality of $\lambda$, and the left- and
right-hand parts by the naturality of $q$. The upper and lower parts
commute by assumption. Thus, the desired inner square commutes when
precomposed by the epimorphism $q_{TX}$, which implies that it
commutes.

For the unit law consider the diagram below:
\[
  \begin{tikzcd}
    F
    \ar{rd}{F\eta}
    \ar[->>]{rrr}{q}
    \ar{ddd}[swap]{\id}
    &&&
    G
    \ar{ld}[swap]{G\eta}
    \ar{ddd}{\id}
    \\
    &
    FT
    \ar{d}[swap]{\lambda}
    \ar[->>]{r}{qT}
    &
    GT
    \ar{d}{\rho}
    \\
    &
    TF
    \ar{r}{Tq}
    &
    TG
    \\
    F
    \ar{ru}{\eta F}
    \ar[->>]{rrr}{q}
    &&&
    G\ar{lu}[swap]{\eta G}
  \end{tikzcd}
\]
The inner square commutes by assumption and the outside trivially
does. The left-hand part commutes by the unit law for $\lambda$, the
lower part by the naturality of $\eta$, and the upper part by the
naturality of $q$. Thus, the desired right-hand part commutes when
precomposed by the epimorphism $q$ at the top, whence it commutes.

Finally, for the multiplication law of $\rho$ we consider the
following diagram
\[
  \begin{tikzcd}
    FTT
    \ar[->>]{rd}{qTT}
    \ar{rr}{\lambda T}
    \ar{ddd}[swap]{F\mu}
    &&
    TFT
    \ar{rr}{T\lambda}
    \ar{d}{TqT}
    &&
    TTF
    \ar{ddd}{\mu F}
    \ar{ld}[swap]{TTq}
    \\
    &
    GTT
    \ar{r}{\rho T}
    \ar{d}[swap]{G\mu}
    &
    TGT
    \ar{r}{T\rho}
    &
    TTG
    \ar{d}{\mu G}
    \\
    &
    GT
    \ar{rr}{\rho}
    &&
    TG
    \\
    FT
    \ar{rrrr}{\lambda}
    \ar[->>]{ru}{qT}
    &&&&
    TF
    \ar{lu}[swap]{Tq}
  \end{tikzcd}
\]
The outside commutes due to the multiplication law for $\lambda$. The
two upper inner parts and the lower one commute by assumption, the
left-hand part commutes by the naturality of $q$, and the right-hand
part commutes by the naturality of $\mu$. Thus, the desired inner
rectangle commutes when precomposed by the epimorphism $qTT$, thus it
commutes.

\subsection*{Proof of \autoref{P:dcpo}}
\begin{proof}
  We first consider $\powufs$.
  Given a family $f_i\colon X \to \powufs Y$ ($i\in I$) of equivariant
  functions we first show that for every $x \in X$ the union $\bigcup_{i\in
    I} f_i(x)$ is uniformly finitely supported by $\supp (x)$. Indeed, given
  $y$ in that union, there is some $i \in I$ such that $y \in
  f_i(x)$. Then we have $\supp (y) \subseteq \supp (f_i(x)) \subseteq
  \supp(x)$, where the first inclusion uses that $f_i(x)$ is uniformly
  finitely supported and the second one uses that $f_i$ is
  equivariant. Thus, we have a function $f\colon X \to \powufs Y$
  given by $f(x) = \bigcup_{i\in I}f_i(x)$. The equivariance of $f$
  easily follows from the equivariance of the $f_i$ and that of
  unions: for every $\pi \in \Perm(\names)$ we have
  \[\textstyle
    \pi \cdot f (x)
    =
    \pi \cdot \bigcup_{i\in I} f_i(x)
    =
    \bigcup_{i\in I} \pi \cdot f_i(x)
    =
    \bigcup_{i\in I} f_i(\pi \cdot x)
    =
    f(\pi \cdot x).
  \]
  Finally, it is clear that $f$ is the join of the $f_i$ in $\Kl\powufs(X,Y)$.

 The proof for $\powfs$ is analogous and only differs
  in one aspect: given $f_i\colon X \to \powfs Y$ ($i\in I$), the
  union $\bigcup_{i\in I} f_i(x)$ is finitely supported since the set
  $\set{f_i(x) : i \in I}$ is supported by $\supp (x)$ due to
  $\supp (f_i(x)) \subseteq \supp (x)$.
  \qed
\end{proof}

\subsection*{Proof of \autoref{P:comm}}

\begin{proof}
  Given nominal sets $X$ and $Y$, an element $x \in X$, and a
  (uniformly) finitely supported subset $S$ of $Y$, the set
  $\set{(x,s) : s \in S}$ is clearly (uniformly) finitely supported by
  $\supp(x) \cup \supp(S)$. This shows that the strength maps
  in~\eqref{eq:pstr} restrict to $\powfs$ and $\powufs$, and these
  maps are also easily seen to be equivariant. Now the validity of all
  the required equational axioms for commutativity are inherited from
  those for the (co-)strength of the power-set functor.
  \qed
\end{proof}

\subsection*{Proof of \autoref{P:abs-dist}}

\begin{rem}\label{R:cong}
  The proof makes use of some points that we mention upfront.
  \begin{enumerate}
  \item\label{R:cong:1} The finitely supported power-set functor
    distributes over the abstraction functor, that is
    $\powfs([\names]X) \cong [\names]\powfs X$ via a natural
    isomorphism $\rho_X$. This follows from Pitts'
    result~\cite[Prop.~4.14]{Pitts13} which shows that the abstraction
    functor preserves exponentials. The proof exhibits a family of
    equivariant isomorphisms
    $\psi_X\colon \powfs([\names]X) \to [\names]\powfs X$ defined by
    $\psi_X(S) = \braket a \set{x : \braket a x \in S}$, where $a$ is
    fresh for $S$.

    We prove that $\psi_X\colon \powfs[\names] X \to [\names]\powfs X$
    is natural in $X$ (the only point not proved by
    Pitts~\cite[Prop.~4.14]{Pitts13}). The naturality of $\rho$ then
    ensues. For any $S$ in $\powufs [\names] X$ we write $S_a$ for the
    set $\set{y : \braket a y \in S}$. Given an equivariant map
    $f\colon X \to Y$ we need to prove that the following square
    commutes
    \[
      \begin{tikzcd}
        \powfs [\names]X
        \ar{r}{\psi_X}
        \ar{d}[swap]{\powfs[\names] f}
        &
        {[\names]} \powfs X
        \ar{d}{[\names]\powfs f}
        \\
        \powfs[\names]Y
        \ar{r}{\psi_Y}
        &
        {[\names]}\powfs Y
      \end{tikzcd}
    \]
    Applying both composites above to a finitely supported subset $S$
    in $\powfs [\names] X$ and unfolding definitions, this boils down
    to showing that
    \[
      \braket a \big([\names] f [S]\big)_a = \braket a f[S_a].
    \]
    Using Pitts~\cite[Lem.~4.3]{Pitts13}, this holds if and only if the
    sets $\big([\names] f[S]\big)_a$ and $f[S_a]$ are equal. This is
    seen from the following chain of equivalences
    \begin{align*}
      y \in f[S_a] 
      &\iff \exists x \in S_a.\, f(x) = y \\
      &\iff \exists x \in X.\, \text{$f(x) = y$ and $\braket a x \in S$}\\
      &\overset{(*)}{\iff} 
      \exists \braket b x \in S.\, \braket b f(x) = \braket a y\\
      &\iff \braket a y\in [\names] f[S]\\
      &\iff y\in \big([\names] f[S]\big)_a
    \end{align*}
    In the equivalence labelled by ($*$) the implication from left to
    right is clear (take $b = a$). For the reverse implication take
    $b \fresh S$ and let $\pi = (a\, b)$. Since $a, b$ are fresh for $S$
    (the former by definition of $\psi_X$) and $\braket b x \in S$, we
    have $\braket a (\pi \cdot x) = \pi \cdot (\braket b x) \in S$. In
    addition, we have that $f(\pi \cdot x) = y$ since
    \[
      \braket a f(\pi \cdot x) = \pi \cdot (\braket b f(x)) = \pi \cdot
      (\braket a y) = \braket a y, 
    \]
    where the last equation holds because $a$ and $b$ are fresh for
    $\braket a y$. 

  \item Even though not needed for the proof below, let us also
    mention that, specializing the description from the proof of
    \cite[Prop.~4.14]{Pitts13} we see that the components of the
    inverse of $\psi$ are given by
    $\rho_X\colon [\names]\powfs X \to \powfs([\names]X)$ defined by
    \[
      \rho_X(\braket a S) = \set{z : z \con b \in \big((\braket a S) \con b\big)},
    \]
    where $b \neq a$ is fresh for $S$. Here, $\con$ is the concretion operator given by $(\braket a x) \con b=x$ if $b = a$ and $(\braket a x) \con b=(a\, b)
    \cdot x$ if $a \neq b$ is fresh for $x$, and undefined otherwise.
    
  \item The abstraction functor is a quotient of the functor
    $F_2 X = \names \times X$ via the natural transformation given by
    the canonical quotient maps
    $q_X\colon \names \times X \to [\names] X$ defined by
    $(a,x) \mapsto \braket a x$ for every $a \in \names$ and
    $x \in X$. Moreover, we know that we have a distributive law of
    $F_2$ over the monad $\powufs$ given by the strength, that is
    $\lambda_X\colon \names \times \powufs X \to \powufs(\names \times
    X)$ is defined by $(a,S) \mapsto \set {(a,s) : s \in S}$.
  \end{enumerate}
\end{rem}
\begin{proof}[\autoref{P:abs-dist}]
  According to \autoref{R:Kleisli}\ref{R:Kleisli:3} and the third
  point above it suffices to exhibit a family of equivariant maps 
  $\rho_X\colon [\names] \powufs X \to \powufs [\names] X$ such that
  the following diagram commutes
  \begin{equation}\label{eq:square}
    \begin{tikzcd}
      \names\times \powufs X
      \ar{r}{\lambda_X}
      \ar[->>]{d}[swap]{q_{\powufs X}}
      &
      \powufs(\names \times X)
      \ar{d}{\powufs q_X}
      \\
      {[\names]\powufs X}
      \ar{r}{\rho_X}
      &
      {\powufs[\names] X}
    \end{tikzcd}
  \end{equation}
  \begin{enumerate}
  \item First, we note that the maps $\psi_X$ in
    \autoref{R:cong}\ref{R:cong:1} restrict to $\powufs$. Indeed,
    given a uniformly finitely supported subset
    $S \subseteq [\names] X$ we see that for every $\braket a x \in S$
    we have $\supp(x) \subseteq \set a \cup \supp(S)$. Thus,
    $\set{x : \braket a x \in S}$ is uniformly finitely supported, and
    therefore so is $\psi_X(S)$.

  \item We prove that the inverse of $\psi_X\colon
    \powufs([\names] X) \to [\names]\powufs X$ is
    \begin{equation}\label{eq:rho-simp}
      \rho_X (\braket a S) = \set{\braket a s : s \in S}.
    \end{equation}
    We clearly have for every $\braket a S$ in $[\names]\powufs X$ that
    \[
      \psi_X \cdot \rho_X (\braket{a} S) = \psi_X\big(\set{\braket a s :
        s \in S}) = \braket a S.
    \]
    Furthermore, given $S$ in $\powufs([\names] X)$, we choose $a$ fresh
    for $S$ and compute
    \[
      \rho_X \cdot \psi_X(S)
      =
      \rho_X\big(\braket a \set{x : \braket a x \in S} \big)
      =
      \set{\braket a x : \braket a x \in S}
      =
      S.
    \]


  \item For the commutativity of the square~\eqref{eq:square} we compute as
    follows for every $(a,S) \in \names \times \powufs X$:
    \begin{align*}
      \powufs q_X \big(\lambda_X(a,S)\big)
      &= \powufs q_X \set{(a,s) : s \in S}
      =
      \set{\braket a s : s \in S}
      =
      \rho_X (\braket a S)
      \\
      &=
      \rho_X\big(q_{\powufs X}(a,S)\big).
    \end{align*}
  \item Finally, we prove that the extension $\barAbs(-)$ ensuing from
    the distributive law~$\rho$ is locally monotone. It maps a nominal
    set $X$ to $[\names]X$ and an equivariant map
    $f\colon X \to \powufs Y$ to
    $\barAbs f = \rho_Y \cdot [\names]f\colon [\names]X \to \powufs
    {[\names]Y}$. Let $f, g\colon X \to \powufs Y$ satisfy
    $f \leqslant g$, which means that $f(x) \subseteq g(x)$ for all
    $x \in X$. Then we clearly have for every
    $\braket a x \in [\names]X$ that
    \[
      \barAbs f(\braket a x)
      =
      \set{\braket a y : y \in f(x)}
      \subseteq
      \set{\braket a y : y \in g(x)}
      =
      \barAbs g(\braket a x).
      \tag*{\qed}
    \]
  \end{enumerate}
\end{proof}

\subsection*{Proof of \autoref{cor:extension-pufs}}

\begin{proof}
  This is shown by structural induction following the grammar in
  \eqref{eq:grammar}. The base case follows from
  \autoref{P:abs-dist} since constant functors and the identity
  clearly canonically extend to locally monotone functors on
  $\Kl\powufs$. For the induction step use
  \autoref{E:ext}\ref{E:ext:2} and~\ref{E:ext:3} and the
  easily established fact that the canonical extension of a coproduct
  or a finite product of functors with a locally monotone extension is
  locally monotone, too. \qed
\end{proof}

\subsection*{Proof of \autoref{T:NOFA-Kl}}

\begin{proof}
  Since the transitions and final states of a NOFA are equivariant,
  the data language $L_c(x)$ accepted by any state $x\in X$ is finitely
  supported by $\supp(x)$, and $L_c(\pi\cdot x)=\pi\cdot L_c(x)$ for all
  $\pi\in \Perm(\names)$. Thus $L_c\colon X\to \powfs(\names^*)$ is a
  well-defined equivariant map. It suffices to show that the square
  below in $\Kl\powfs$ commutes; then $\tr_c=L_c$ by uniqueness of
  $\tr_c$.
  \[
    \begin{tikzcd}[column sep = 30, row sep = 15]
      X \ar{d}[swap]{c} \ar{r}{L_c}
      &
      \mu F
      \ar{d}{J\ini^{-1}}
      \\
      \barF X
      \ar{r}{\bar F L_c}
      &
      \barF(\mu F)
    \end{tikzcd}
    \qquad = \qquad
    \begin{tikzcd}[column sep = 35, row sep = 18]
      X \ar{d}[swap]{c} \ar{r}{L_c}
      &
      \names^*
      \ar{d}{J\ini^{-1}}
      \\
      1+\names\times X
      \ar{r}{\bar F L_c}
      &
      1+\names\times \names^*
    \end{tikzcd}
\]
Here the map $\bar F L_c\colon 1+\names\times X\to
\powfs(1+\names\times \names^*)$ is given by
\[
  \bar F L_c(\ast)=\{\ast\}
  \qquad\text{and}\qquad
  \bar F L_c(a,x) = \{ (a,w) : w\in L_c(x) \},
\]
where $1=\{\ast\}$. Denoting Kleisli composition by $\bullet$, for every $x\in X$ we have
\begin{align*}
\ast \in \barF L_c\bullet c(x) & \iff \ast \in c(x) \\
&\iff \text{$x$ is a final state} \\
&\iff \epsilon \in L_c(x) \\
&\iff \ast\in J\iota^{-1}\bullet L_c(x).
\end{align*}
Moreover, for $a\in \names$ and $w\in \names^*$ we compute
\begin{align*} (a,w)\in  \barF L_c\bullet c(x) &\iff \exists y\in X: (a,y)\in c(x) \wedge w\in L_c(y) \\
&\iff \exists y\in X: \text{$x\xto{a} y$ and $y$ accepts $w$} \\
&\iff aw\in L_c(x) \\
&\iff (a,w)\in J\iota^{-1}\bullet L_c(x).
\end{align*}
Thus $\barF L_c\bullet c=J\iota^{-1}\bullet L_c$ as claimed.
\qed
\end{proof}

\subsection*{Proof of \autoref{T:RNNA-Kl}}

\begin{proof}
  Note first that the bar language $L_c(x)$ accepted by any state
  $x\in X$ is uniformly finitely supported by $\supp(x)$, see
  \cite[Cor.~5.5]{SchroderEA17}. Moreover, by equivariance of
  transitions and final states, we have $L_c(\pi\cdot x)=\pi\cdot L_c(x)$
  for all $\pi\in \Perm(\names)$. Thus
  $L_c\colon X\to \powufs(\names^*)$ is a well-defined equivariant
  map. It suffices to show that the square below in $\Kl\powufs$
  commutes; then $\tr_c=L_c$ by uniqueness of $\tr_c$.
  \[
    \begin{tikzcd}[column sep = 30, row sep = 18]
      X \ar{d}[swap]{c} \ar{r}{L_c}
      &
      \mu F
      \ar{d}{J\ini^{-1}}
      \\
      \barF X
      \ar{r}{\bar F L_c}
      &
      \barF(\mu F)
    \end{tikzcd}
    \quad = \hspace{2pt} 
    \begin{tikzcd}[column sep = 30, row sep = 18]
      X \ar{d}[swap]{c} \ar{r}{L_c}
      &
      \barstr
      \ar{d}{J\ini^{-1}}
      \\
      \mathllap{1+\names}\times X+[\names]X
      \ar{r}{\barF L_c}
      &
      1+\names\times \barstr + [\names]\barstr
    \end{tikzcd}
  \]
  Thus let $x\in X$. Denoting Kleisli composition by $\bullet$, the
  equivalences
  \begin{align*}
    \ast \in \barF L_c\bullet c(x)
    & \iff \ast\in J\iota^{-1}\bullet L_c(x)
    \\
    (a,[w]_\alpha)\in  \barF L_c\bullet c(x)
    &\iff (a,[w]_\alpha)\in J\iota^{-1}\bullet L_c(x).
  \end{align*}
  for $a\in \names$ and $w\in \barA^*$ are established as in the proof
  of \autoref{T:NOFA-Kl}. Moreover,
  \begin{align*} 
    & \braket{a}[w]_\alpha \in \bar F L_c\bullet c(x)
    \\
    \iff &
    \exists b\in\names, v\in \names^*, y\in X:  \braket{b}y\in c(x),\, [v]_\alpha\in L_c(y),\, \braket{a}[w]_\alpha = \braket{b}[v]_\alpha \\
\iff &  \exists b\in\names, v\in \names^*, y\in X:  x\xto{\scriptnew b} y,\, [v]_\alpha\in L_c(y),\, \braket{a}[w]_\alpha = \braket{b}[v]_\alpha \\
\iff & \exists b\in\names, v\in \names^*: [\newletter bv]_\alpha\in L_c(y),\, \braket{a}[w]_\alpha = \braket{b}[v]_\alpha \\
\iff & \exists b\in\names, v\in \names^*: [\newletter bv]_\alpha\in L_c(y),\, [\newletter a w]_\alpha=[\newletter b v]_\alpha \\
\iff & [\newletter aw]_\alpha\in L_c(y) \\
\iff & \braket{a}[w]_\alpha \in J\iota^{-1}\bullet L_c(x).
\end{align*}
This proves $\barF L_c\bullet c=J\iota^{-1}\bullet L_c$ as claimed.
\qed
\end{proof}

\section{Details for \autoref{S:lang}}

\subsection*{Proof of \autoref{P:adj-trans}}

\begin{proof}
  Recall that the adjoint transpose of a morphism $f\colon X \to RY$,
  that is its image under the inverse of the natural isomorphism
  $\C(LX, Y) \cong \C(X, RY)$, is given by
  \[
    LX \xra{Lf} LRY \xra{\eps_Y} Y.
  \]
  This implies that $\rho\colon TR \to RT$ is uniquely determined by
  $\eps \cdot L\rho$. Hence, according to the statement of our
  proposition $\rho$ is determined by the commutativity of the
  square below:
  \[
    \begin{tikzcd}
      LTR
      \ar{r}{L\rho}
      \ar{d}[swap]{\lambda R}
      &
      LRT
      \ar{d}{\eps T}
      \\
      TLR
      \ar{r}{T\eps}
      &
      T
    \end{tikzcd}
  \]
  
  We verify that the two properties of a distributive law for $\rho$
  follow from those of $\lambda$ by taking adjoint transposes. That
  is, it suffices to show that the desired diagram commutes when we
  apply $L$ to them and postcompose with $\eps T$. For the unit law
  $\rho \cdot \eta R = R\eta$ we obtain the commutative diagram below:
  \[
    \begin{tikzcd}
      LR
      \ar{r}{L\eta R}
      \ar{dd}[swap]{LR\eta}
      \ar{rd}[inner sep=0]{\eta LR}
      \ar{rdd}[swap]{\eps}
      &
      LTR
      \ar{r}{L\rho}
      \ar{d}{\lambda R}
      &
      LRT
      \ar{dd}{\eps T}
      \\
      &
      TLR
      \ar{rd}{T\eps}
      \\
      LRT
      \ar[shiftarr = {yshift=-20}]{rr}{\eps T}
      &
      \Id
      \ar{r}{\eta}
      &
      T
    \end{tikzcd}
  \]
  The right-hand part commutes by the definition of $\rho$, the
  left-hand triangle by the unit law for $\lambda$, the middle
  triangle by the naturality of $\eta$ and the remaining lower part by
  the naturality of $\eps$.

  For the multiplication law $\rho \cdot \mu R = R\mu
  \cdot \rho T \cdot T\rho$ we obtain the diagram below:
  \[
    \begin{tikzcd}
      TLTR
      \ar{rr}{T\lambda R}
      \ar{dd}[swap]{T\lambda R}
      \ar{rrd}{TL\rho}
      &&
      TTLR
      \ar{rr}{TT\eps}
      &&
      TT
      \ar{dddd}{\mu}
      \\
      &&
      TLRT
      \ar{rru}{T\eps T}
      \\
      TTLR
      \ar{dd}[swap]{\mu LR}
      &
      LTTR
      \ar{r}{LT\rho}
      \ar{d}[swap]{L\mu R}
      \ar{luu}[swap]{\lambda TR}
      &
      LTRT
      \ar{r}{L\rho T}
      \ar{u}{\lambda RT}
      &
      LRTT
      \ar{d}{LR\mu}
      \ar{ruu}{\eps TT}
      \\
      &
      LTR
      \ar{rr}{L\rho}
      \ar{ld}{\lambda R}
      &&
      LRT
      \ar{rd}{\eps T}
      \\
      TLR
      \ar{rrrr}{T\eps}
      &&&&
      T
    \end{tikzcd}
  \]
  First, to see that the outside commutes remove $T\lambda R$ at the
  beginning of both paths and use the naturality of $\mu$. The inner
  part commutes as follows: the upper triangle commutes by the
  definition of $\rho$, the left-hand part below it by the naturality
  of $\lambda$, the part to its right commutes by the definition of
  $\rho$, the left-hand part commutes by the multiplication law for
  $\lambda$, the right-hand part commutes by the naturality of $\eps$,
  and the lower part by the definition of $\rho$. Thus, the middle
  rectangle commutes when postcomposed by $\eps T$, which is what we
  need to prove, and so we are done.
  \qed
\end{proof}

\subsection*{Details for \autoref{R:eps}\ref{R:eps:2}}

We will prove that $h^\sharp\colon E(X,c) \to E(Y,d)$ is a
homomorphism of $\hatG$-coalgebras and show that $E$ is a
functor. First, we collect some properties of Kleisli extensions and
homomorphic extensions (\autoref{R:hom-ext}):
\begin{rem}
  \begin{enumerate}
  \item The action of $\barF\colon \Kl T \to \Kl T$ on a morphism
    $f\colon X \to TY$ is given using the corresponding distributive
    law $\lambda\colon FT \to TF$ as follows:
    \begin{equation}\label{eq:ext}
      \barF (f) = \big(FX \xra{Ff} FTY \xra{\lambda_Y} TFY\big).
    \end{equation}
  \item Given a morphism $f\colon X \to
    TY$ in $\C$ the homomorphic extension fulfils
    \begin{equation}\label{eq:sharp}
      f^\sharp = \big(TX \xra{Tf} TTY \xra{\mu_Y} TY\big).
    \end{equation}
    
  \item Given an algebra $(A,a)$ for $T$, a morphism $f\colon X \to TY$ and a morphism
    $h\colon (A,a) \to (B,b)$ in $\EM T$ we have
    \begin{equation}\label{diag:sharp}
      \begin{tikzcd}
        TX
        \ar{r}{f^\sharp}
        \ar{rd}[swap,inner sep=1]{(h \cdot f)^\sharp}
        &
        A
        \ar{d}{h}\\
        &
        B
      \end{tikzcd}
    \end{equation}
    Indeed, both paths are morphisms in $\EM T$ which agree when
    precomposed with $\eta_X\colon X \to TX$. The universal property
    of the free algebra $(TX, \mu_X)$ thus yields the commutativity of
    the above triangle.

  \item Restricting homomorphic extensions to free algebras yields the
    extension operation from the presentation of $T$ as a Kleisli
    triple. That means that for every morphisms $f\colon X \to TY$ and
    $g\colon Y \to TZ$ we have
    \begin{equation}\label{eq:kl-laws}
      f^\sharp \cdot \eta_X = f,
      \qquad
      \eta_X^\sharp = \id_{TX},
      \qquad
      g^\sharp \cdot f^\sharp = (g^\sharp \cdot f)^\sharp.
    \end{equation}
  \end{enumerate}
\end{rem}

We prove the desired result that the following square commutes in~$\EM T$:
\[
  \begin{tikzcd}[column sep = 30]
    TX
    \ar{d}[swap]{(\eps_X \cdot c)^\sharp}
    \ar{r}{h^\sharp}
    &
    TY
    \ar{d}{(\eps_Y \cdot d)^\sharp}
    \\
    \hatG TX
    \ar{r}{\hatG h^\sharp}
    &
    \hatG TY
  \end{tikzcd}
\]
It suffices to show that the diagram commutes (in $\C$) when
precomposed by the universal morphism $\eta_X\colon X \to TX$ of the
free Eilenberg-Moore algebra $(TX, \mu_X)$. Using $h^\sharp \cdot
\eta_X = h$ and similarly for $\eps_X\cdot c$ it is thus our task to
show that the outside of the following diagram commutes:
\[
  \begin{tikzcd}[column sep = 30]
    X
    \ar{d}[swap]{c}
    \ar{rrr}{h}
    &&&
    TY
    \ar{ld}{Td}
    \ar{dd}{d^\sharp}
    \ar[shiftarr = {xshift=25}]{ddd}{(\eps_Y \cdot d)^\sharp}
    \\
    TFX
    \ar{rd}[swap, inner sep=1, near end]{TFh}
    \ar{dd}[swap]{\eps_X}
    \ar[bend left = 10]{rrd}{T\barF h}
    &&
    TTFY
    \ar{rd}{\mu_{FY}}
    \\
    &
    TFTY
    \ar{r}{T\lambda_Y}
    \ar{d}{\eps_{TY}}
    &
    TTFY
    \ar{r}{\mu_{FY}}
    &
    TFY
    \ar{d}{\eps_Y}
    \\
    GTX
    \ar{r}{GTh}
    \ar[shiftarr = {yshift=-15}]{rrr}{Gh^\sharp}
    &
    GTTY
    \ar{rr}{G\mu_Y}
    &&
    GTY
  \end{tikzcd}
\]
Indeed, the upper inner part commutes since $h$ is a homomorphism of
coalgebras for $\barF$ (unfolding the definition of the composition in
$\Kl T$; see \autoref{R:Kleisli}\ref{R:Kleisli:1}). The triangle below
it commutes due to~\eqref{eq:ext}, and the triangle on the right as
well as the lowest part by~\eqref{eq:sharp}. The right-hand part
commutes by~\eqref{diag:sharp} using that $\eps_Y$ is a morphism in
$\EM T$ (\autoref{R:eps}\ref{R:eps:1}). Finally, the
lower left-hand part commutes by the naturality of $\eps$, and the
lower right-hand one due to the left-hand law in~\eqref{diag:eps}.

Establishing functoriality of $E$ directly is straightforward using
the Kleisli laws~\eqref{eq:kl-laws}. More conceptually, $E$ is a
clearly defined as a lifting of the canonical comparison functor
$K\colon \Kl T \to \EM T$, so functoriality of $E$ is obvious.  \qed

\subsection*{Proof of \autoref{P:relation}}

\begin{proof}
  By \autoref{R:eps}\ref{R:eps:2}, we know that
  $E\tr_c = \tr_c^\sharp\colon TX \to T(\mu F)$ is a homomorphism of
  coalgebras for $\hatG$ from the $\hatG$-coalgebra $E(X,c)$ viz.~the
  determinization of the coalgebra
  \[
    X \xra{c} TFX \xra {\eps_X} GTX
  \]
  to $E(\mu F, J\ini^{-1})$. Since $\nu G$ carries the terminal
  $\hatG$ coalgebra, we have $e \cdot \tr_c^\sharp = h$ where
  \[
    E(X,c) =  (TX, (\eps_X \cdot c)^\sharp) \xra{h} \nu G
  \]
  is the unique $\hatG$-coalgebra homomorphism. Precomposing by
  $\eta_X$ we obtain
  \[
    \lan(\eps_X \cdot c) = h \cdot \eta_X = e \cdot \tr_c^\sharp
    \cdot \eta_X = e \cdot \tr_c. \tag*{\qed}
  \]
\end{proof}

\subsection*{Remark on \autoref{P:lift-abst}}

\medskip\noindent 

 \noindent It is not hard to see that the components of the distributive law arising from the
  proof of \autoref{P:lift-abst} are the natural isomorphisms
  $\psi_X \colon \powufs([\names] X) \to [\names]\powufs X$ given by
  $\psi_X = \braket{a}\{ x : \braket{a}x\in S\}$, where~$a$ is fresh
  for~$S$ (see \autoref{R:cong}). 

  \smnote{The following can go to the appendix.}%
  To see this recall from Pitts~\cite[Thm.~4.12]{Pitts13} that the
  counit of the adjunction $\names * (-) \dashv [\names](-)$ is given
  by
  \[
    \eps_X\colon \names * [\names] X \to X,
    \qquad
    (a, \braket b x) \mapsto (a\, b)\cdot x. 
  \]
  Moreover, the adjoint transpose of an equivariant map $f\colon
  \names * X \to Y$ is 
  \[
    \hat f\colon X \to [\names] Y,
    \qquad
    x \mapsto \braket a f(a,x),
  \]
  where $a$ is fresh for $x$.%
  \smnote{Please check this carefully; I am always confused by Pitts' \texttt{fresh}-notation.}
  From the proof of \autoref{P:adj-trans} we see that the distributive
  law of $\powufs$ over $[\names](-)$ is given by the adjoint
  transposes of
  \begin{equation}\label{eq:eps-s}
    \begin{tikzcd}[row sep = -1,column sep = 15]
      \names * \powufs([\names] X)
      \ar{r}{s_{\names,X}}
      &
      \powufs(\names * [\names]X)
      \ar{r}{\powufs \eps_X}
      &
      \powufs X
      \\      
      (a,S)
      \ar[|->]{r}
      &
      \set{(a,\braket b x) : \braket b x \in S}
      \ar[|->]{r}
      &
      \{(a\, b)\cdot x : \mathrlap{\braket b x \in S\}.}
    \end{tikzcd}
    \qquad
  \end{equation}
  The nominal set on the right above is easily seen to be equal to
  $\set{x : \braket a x \in S}$. Indeed, for ``$\subseteq$'' note that
  $\braket b x = \braket a ((a\,b)\cdot x)$, and ``$\supseteq$'' is
  obvious since for $a=b$ we have $(a\, b) \cdot x = x$.
  Thus, the adjoint transpose of the equivariant map
  in~\eqref{eq:eps-s} is~$\psi_X$ as desired. 

\subsection*{Proof of \autoref{prop:ter}}

\begin{proof}
\begin{enumerate}
\item\label{prop:ter:1} Let $\cpowfs\colon \Nom \to \Nom^\opp$ denote the contravariant finitely supported power-set functor defined by
$\cpowfs X = \powfs X$ and $\cpowfs f\colon A \mapsto f^{-1}[A]$. Note that $\cpowfs$ is naturally isomorphic to the exponentiation functor $X\mapsto 2^X$, so it has the right adjoint $\cpowfs^\opp\colon \Nom^\opp\to \Nom$. Moreover, $\cpowfs$ commutes with the abstraction functor: for each nominal set $X$ we have the bijection
\[ \phi_X\colon \cpowfs([\names]X)\xto{\cong} [\names](\cpowfs X) \]
defined by $\phi_X(S)=\braket{a}\{x : \braket{a}x\in S\}$ where $a$ is fresh for $S$, cf. \autoref{R:cong}\ref{R:cong:1}. It is natural in $X$, that is, the following square commutes for every $f\colon Y\to X$:
\[  
  \begin{tikzcd}[column sep = 40]
    \cpowfs({[\names]X})  \ar{d}[swap]{\phi_X}  \ar{r}{\cpowfs([\names]
      f)}
    &
    \cpowfs({[\names]Y})
    \ar{d}{\phi_Y}
    \\
    {[\names]}\cpowfs X
    \ar{r}{[\names]\cpowfs f}
    &
    {[\names]}\cpowfs Y
\end{tikzcd}
\]
To see this, let $S\in \cpowfs([\names]X)$ and pick $a\in \names$ fresh for $S$. Then $a$ is also fresh for $\cpowfs([\names]f)(S)=([\names]f)^{-1} [S]$, and so 
\begin{align*}
\phi_Y(\cpowfs([\names]f)(S)) &= \phi_Y(([\names]f)^{-1}[S]) \\
&= \braket{a} \{ y\in Y: \braket{a}y\in ([\names]f)^{-1}[S]  \} \\
&= \braket{a} \{ y\in Y: \braket{a}f(y) \in S \} \\
&= \braket{a} f^{-1}( \{ x\in X: \braket{a}x\in S \} ) \\
&= \braket{a} \cpowfs f( \{ x\in X: \braket{a}x\in S  \}) \\
&= [\names](\cpowfs f) ( \braket{a} \{ x\in X: \braket{a}x\in S  \} ) \\
&= [\names](\cpowfs f) ( \phi_X(S)  ).
\end{align*}

\item\label{prop:ter:2} Let $FX=1+\names\times X+ [\names]X$ and $GX=2\times X^\names \times [\names]X$.
Then the square
\begin{equation}\label{eq:qfs-square}
\begin{tikzcd}
\Nom  \ar{r}{\cpowfs} \ar{d}[swap]{F} & \Nom^\opp \ar{d}{G^\opp} \\
\Nom \ar{r}{\cpowfs} & \Nom^\opp
\end{tikzcd}
\end{equation}
commutes up to natural isomorphism: for every $X\in\Nom$ we have
\begin{align*} 
\cpowfs FX &= \cpowfs(1+\names\times X + [\names]X) \\
& \cong 2\times \cpowfs(\names\times X) \times \cpowfs([\names]X) \\
& \cong 2\times (\cpowfs X)^\names \times [\names](\cpowfs X) \\
& = G^\opp \cpowfs X 
\end{align*}
where the penultimate step uses the laws of exponentiation and the
natural isomorphism $\phi_X$ from \autoref{prop:ter:1} of the proof.
\item By \cite[Thm.~2.14]{HermidaJ98}, commutativity of the square \eqref{eq:qfs-square} up to natural isomorphism implies the left adjoint $\cpowfs$ lifts to a left adjoint 
\[ \overline{\cpowfs}\colon \Alg{F}\to \Alg{G^\opp}\cong (\Coalg{G})^\opp \]
mapping an $F$-algebra $FA\xto{\alpha} A$ to the $G$-coalgebra
\[ \powfs A = \cpowfs A \xto{\cpowfs \alpha} \cpowfs FA \cong G^\opp \cpowfs A = G\powfs A. \]
In particular, since the left adjoint $\overline{\cpowfs}$ preserves colimits (whence initial objects), we see that the terminal coalgebra for $G$ is given by 
\begin{align*}
  \powfs (\barstr) = \cpowfs (\barstr) \xto{\cpowfs \iota} \cpowfs
  F(\barstr) &\cong G^\opp \cpowfs (\barstr) \\
  &= G\powfs (\barstr),
\end{align*}
where $F(\barstr)\xto{\iota} \barstr$ is the initial algebra of $F$
(see \autoref{prop:ini}). By definition of $\iota$ and of the
isomorphism $\cpowfs F\cong G^\opp \cpowfs$ in \autoref{prop:ter:2}
above, this yields precisely the desired coalgebra structure $\ter$ on $\powfs(\barstr)$.
\qed
\end{enumerate}
\end{proof}

\subsection*{Proof of \autoref{L:eps-ext}}

\begin{proof}
  We verify that the two laws in~\eqref{diag:eps} hold.
  \begin{enumerate}
  \item For verification of the left-hand law first note that the distributive law
    \[
      \lambda_X\colon
      1 + \names \times \powufs X + [\names](\powufs X)
      \to
      \powufs (1 + \names \times X + [\names] X)
    \]
    corresponding to the canonical lifting $\barF$ is given on the three
    coproduct components of its domain by 
    \[
      \lambda_X(*) = \set{*}, \quad
      \lambda(a, S) = \set{(a,x) : x \in S}, \quad
      \lambda(\braket a S) = \set{\braket a s : s \in S};
    \]
    indeed, this can be gleaned from \autoref{E:ext} and the proof of
    \autoref{P:abs-dist}.
    
    Now given
    $\S \in TFTX = \powufs(1 + \names \times \powufs X + [\names]
    (\powufs X))$, we apply the definitions of $\lambda_X$ and $\mu_X$
    we see that the set $\mu_{FX} \cdot \powufs\lambda_X(\S)$ is the
    union
    \begin{equation}\label{eq:aux}
      \set{* : * \in \S} \cup \set{(a,x) : (a,S) \in \S, x \in S } \cup \set{\braket a x :
        \braket a S \in \S, x \in S};
    \end{equation}
    the first set means that $*$ is contained in the union iff
    $* \in \S$. This is mapped by $\eps_X$ to the triple
    \begin{equation}\label{eq:aux-1.5}
      (b, a\mapsto S_a, S_{\scriptnew a}) \in G\powufs X,
    \end{equation}
    where $b = 1$ iff $* \in \S$, $S_a = \set{x : x \in S, (a,S) \in
      \S}$ and $\braket a S_{\scriptnew a} = \set {x : x \in S, \braket a
      S \in \S}$ with $a$ fresh for the right-hand set in~\eqref{eq:aux}
    above. Following the lower path in the left-hand diagram
    from~\eqref{diag:eps} we have by the definition of $\eps$ that
    \begin{equation}\label{eq:aux-2}
      \eps_{\powufs X}(\S) = (b, a \mapsto \S_a, \S_{\scriptnew a}),
    \end{equation}
    where $b = 1$ iff $* \in \S$, $\S_a = \set{S : (a,S) \in \S}$ and
    $\S_{\scriptnew a} = \braket a \set{S: \braket a S \in \S}$. It is
    clear that $G\mu_X$ maps this to the triple in~\eqref{eq:aux-1.5}.%
    \smnote{@all: Mir scheint, das hier noch was mit freshness von $a$
      argumentiert werden muss; bitte prüft das nochmal genau nach.}
    
  \item We verify the right-hand law in~\eqref{diag:eps}. Consider
    the product projections
    \[
      \begin{tikzcd}
        &
        2 \times X^\names \times [\names]X
        \ar{ld}[swap]{p_0}
        \ar{d}{p_1}
        \ar{rd}{p_2}
        \\
        2 & X^\names & {[\names] X}
      \end{tikzcd}
    \]
    The distributive law
    \[
      \rho_X\colon \powufs (2 \times X^\names \times [\names] X)
      \to 2 \times (\powufs X)^\names \times [\names] (\powufs X)
    \]
    corresponding to the canonical lifting $\hatG$ is given by 
    \[
      \rho_X(S) = (b, a \mapsto S_a, S_{\scriptnew a}),
    \]
    where $b = 1$ iff $1\in p_0[S]$, $S_a = \set{f(a) : f \in p_1[S]}$
    and $S_{\scriptnew a} = \braket a \set{s : \braket a s \in p_2[S]}$.
    This can be extracted from \autoref{E:lift} and the proof of
    \autoref{P:lift-abst}.

    Now given $\S \in TTFX = \powufs\powufs(1 + \names \times X +
    [\names] X)$, then following the upper path of the desired diagram
    we obtain
    \begin{equation}\label{eq:aux-3}
      \eps_X \cdot \mu_{FX} (\S) = \eps_X(\textstyle\bigcup \S) = (b, a \mapsto
      \S_a, \S_{\scriptnew a}),
    \end{equation}
    where $b = 1$ iff $* \in S$ for some $S \in \S$, $\S_a = \set{x :
      (a,x) \in S \in \S}$ and $\S_{\scriptnew a} = \braket a \set{x : \braket a
      x \in S\in\S}$ for an $a$ which is fresh for $\bigcup \S$. 
    
    Now let us consider the lower path in the desired diagram: we have
    \[
      \powufs\eps_X (\S) = \set{(b_S, a \mapsto S_a, S_{\scriptnew a})
        : S \in \S} =: U
    \]
    with $b_S$, $S_a$ and $S_{\scriptnew a}$ as
    in~\autoref{R:epsX}. Applying $\rho_{\powufs X}$ to this set $U$
    we obtain
    \[
      \rho_{\powufs X}\cdot \powufs\eps_X (\S) = (b, a\mapsto U_a,
      U_{\scriptnew a}),
    \]
    where $b = 1$ iff $* \in p_0[U]$, which holds iff  $* \in \bigcup
    \S$, and we have
    \begin{align*}
      U_a &= \set{f(a) : f \in p_1[U]} = \set{S_a : S \in \S}, \\
      U_{\scriptnew a} &= \braket a \set{S :
        \braket a S \in p_2[U]} = \braket a \big\{\set{s : \braket a s
          \in S} : S \in \S\big\}.
    \end{align*}
    It is not difficult to see that $G\mu_X$ maps this triple to the
    one in~\eqref{eq:aux-3}.%
    \smnote{Auch hier versteckt sich noch ein fehlendes freshness
      Argument, wenn man die $S_{\scriptnew s}$ expandiert; bitte nachprüfen!}
    This completes the proof. 
    \qed
  \end{enumerate}
\end{proof}

\subsection*{Proof of \autoref{L:incl}}

\begin{rem}\label{R:free-ext}
  Note that for every equivariant map $f\colon X \to \powufs Y$ the
  free extension $f^\sharp\colon \powufs X \to \powufs Y$ is given by 
  $f^\sharp(S) = \bigcup_{s \in S} f(s)$. 
\end{rem}
\begin{proof}
  To see this it suffices to prove
  that for the functor $E\colon \Coalg \barF \to \Coalg \hatG$ from
  \autoref{R:eps}\ref{R:eps:2} the structure of the coalgebra
  \[
    E\big(\mu F \xra{J\ini^{-1}} \powufs(F(\barstr))\big)
    =
    \big(\powufs(\barstr)
    \xra{(\eps_{\bar\names^*/{=_\alpha}} \cdot J\ini^{-1})^\sharp}
    G(\powufs(\barstr))\big)
  \]
  acts like like the coalgebra structure on $\nu G = \barlang$ in~\autoref{prop:ter}.
  From~\eqref{eq:iniF}, \autoref{C:ini}, and the definition of
  $J\colon \Nom \to \Kl\powufs$ we see that
  \[
    J\ini^{-1}([w]_\alpha) =
    \begin{cases}
      \set{*} & \text{if $w = \eps$}, \\
      \set{(a, [v]_\alpha)} & \text{if $w = av$}, \\
      \set{\braket a [v]_\alpha} & \text{if $w = \newletter a v$}.
    \end{cases}
  \]
  We now compose this map with the component of
  $\eps\colon \powufs F \to G \powufs$ in~\autoref{R:epsX} for
  $X = \barstr$, and then freely extend from $\mu F = \barstr$ to
  $\powufs(\mu F) = \barlang$ using $(-)^\sharp$
  (\autoref{R:free-ext}). This clearly yields the desired
  result.
  \qed
\end{proof}

\end{document}
